\documentclass[conference,compsoc]{IEEEtran}
\usepackage{algorithm}
\usepackage{algpseudocode}
\usepackage{algorithmicx}
\usepackage{amsmath}
\usepackage{amsthm}
\usepackage{amssymb}
\usepackage{amsfonts}
\usepackage{subfigure}
\usepackage{bbm}
\usepackage{ulem}

\newtheorem{thm}{Theorem}
\newtheorem{defi}{Definition}

\newtheorem{lem}{Lemma}
\newtheorem{cor}{Corollary}
\newtheorem{rmk}{Remark}

\usepackage{tikz}
\usepackage{amsmath}
\usepackage{hyperref}

\newcommand{\wz}[1]{{\textcolor{blue}{[wz: #1]}}}
\newcommand{\bj}[1]{{\textcolor{red}{[bj: #1]}}}

\ifCLASSINFOpdf
\else
\fi
\begin{document}
%

\title{When Focus Enhances Utility: Target Range LDP Frequency Estimation and Unknown Item Discovery}



%

\author{\IEEEauthorblockN{Bo Jiang\IEEEauthorrefmark{1}\IEEEauthorrefmark{2}, Wanrong Zhang\IEEEauthorrefmark{1}\IEEEauthorrefmark{2},
Donghang Lu\IEEEauthorrefmark{1}\IEEEauthorrefmark{2}, Jian Du\IEEEauthorrefmark{1}, and
    Qiang Yan\IEEEauthorrefmark{1}.}
\IEEEauthorblockA{\IEEEauthorrefmark{1}
TikTok Inc, \IEEEauthorrefmark{2} Equal Contribution.\\ Email: \{bojiang, wanrongzhang, Donghang.lu, jian.du, yanqiang.mr\}@tiktok.com}}



\maketitle

\begin{abstract}
Local Differential Privacy (LDP) protocols enable the collection of randomized client messages for data analysis, without the necessity of a trusted data curator. Such protocols have been successfully deployed in real-world scenarios by major tech companies like Google, Apple, and Microsoft. In this paper, we propose a Generalized Count Mean Sketch (GCMS) protocol that captures many existing frequency estimation protocols. Our method significantly improves the three-way trade-offs between communication, privacy, and accuracy. We also introduce a general utility analysis framework that enables optimizing parameter designs. {Based on that, we propose an Optimal Count Mean Sketch (OCMS) framework that minimizes the variance for collecting items with targeted frequencies.} Moreover, we present a novel protocol for collecting data within unknown domain, as our frequency estimation protocols only work effectively with known data domain. Leveraging the stability-based histogram technique alongside the Encryption-Shuffling-Analysis (ESA) framework, our approach employs an auxiliary server to construct histograms without accessing original data messages. This protocol achieves accuracy akin to the central DP model while offering local-like privacy guarantees and substantially lowering computational costs.

\end{abstract}


%

\section{Introduction}

Differential Privacy (DP) is a mathematical definition of privacy that provides strong worst-case privacy guarantees for individuals within a dataset, while enabling data analysis. While numerous differentially private algorithms have been created for various analysis tasks, most high-profile real-world applications focus on the Local Differential Privacy (LDP) model. The LDP model is particularly attractive because it does not require trust in a central data curator; instead, it ensures that data is obfuscated at the source, before being collected. Some successful deployments in major tech companies include Google's RAPPOR \cite{erlingsson2014rappor,45382}, a protocol integrated into Chrome, to collect web browsing behavior; Apple \cite{apple2017learning} collecting type history and emojis; and Microsoft \cite{10.5555/3294996.3295115} collecting telemetry across millions of devices.

Frequency estimation, where the goal is to estimate the occurrence of items within a dataset, is a common task in data analytics that has seen various implementations under the LDP model. The process for achieving this usually has several steps: clients first encode their responses (inputs) into a designated format, then perturb these encoded values to generate outputs. These outputs are sent to an aggregator, who collects and decodes the values to estimate the number of clients associated with each specific input. 
 One notable framework is the ``pure" Local Differential Privacy (LDP) framework \cite{wang2017locally}, which includes many existing protocols (e.g., Basic RAPPOR) as special cases. This enables us to precisely analyze and compare the accuracy of different protocols and also offers optimal parameters in the randomization step. However, their framework has limitations in scenarios requiring extensive data collection, particularly within large domains.
Some protocols have incorporated with hash functions to encode data into a more manageable domain size. The hash functions, on one hand, reduce communication cost and enable encoding data from unknown or unlimited domains; on the other hand, they introduce collisions during encoding, where multiple items are mapped to the same hash value. These collisions make the raw data values hard to distinguish, leading to reduced utility in terms of the accuracy of aggregation. To counter the utility reduction caused by collisions, RAPPOR adopts client cohorts and extra LASSO clustering in aggregation \cite{Rappor}. However, it is not straightforward to derive an error bound for utility. Comparably, Apple's Count Mean Sketch (CMS) is more favorable, as it utilizes a statistical data structure to average out other values in the collision and features a closed-form variance expression for the aggregation. However, Apple's CMS has deterministic perturbation/aggregation parameters which are not optimized to balance the accuracy and communication tradeoff. In this paper, we propose a Generalized Count-Mean-Sketch (GCMS) protocol that builds upon and extends the foundational ideas of Apple's CMS.

Our proposed mechanism aims to improve the existing framework by accounting for both the randomness inherent in truthful responses and the complications arising from hash collisions. By doing so, we enhance the accuracy of frequency estimation under the LDP model. One of the key advantages of our GCMS protocol is its ability to achieve a significantly smaller communication cost in terms of plaintext length while maintaining the same level of privacy protection as the original CMS. Furthermore, we provide a utility guarantee for our general CMS protocol, enabling the optimal design of parameters tailored to specific item frequency regimes. This ensures that our protocol is not only more efficient in terms of communication cost but also adaptable to various data collection scenarios to further boost the estimation accuracy. 

To further amplify privacy, we present our protocol within the Encryption-Shuffling-Analysis (ESA) framework. Each client's encoded and randomized output is encrypted, and these encrypted outputs are then shuffled by a trusted shuffler, which involves randomly permuting the data to break the direct link between clients and their responses. This randomness in permutation significantly amplifies privacy by making it even harder to trace any response back to an individual client, enhancing an $\epsilon$-LDP randomizer to $\mathcal{O}\left(1-e^{-\epsilon_0}\sqrt{\frac{e^{\epsilon_0}(1/\delta)}{n}}\right)$-DP.
The server then aggregates the shuffled and encrypted data and performs the analysis. 

Note that these frequency estimation protocols only work effectively when the data domain is known. {There are many scenarios where the server does not have full knowledge of the clients' input dictionary. For example, in word typing, new words are constantly being invented alongside the existing dictionary. Another example is URL collection: whenever new content is uploaded, a new URL is generated. } 
To address this limitation, we also offer a protocol that collects new data within an unknown domain. Few work tackles this problem, but they all come with significant computing/communication cost. Fanti et al. \cite{45382} propose an algorithm using RAPPOR that segments the domain, releases segments, and reconstructs data using an EM algorithm, but faces high computational costs. The private sequence fragment puzzle method  \cite{apple2017learning}, based on CMS, also segments and reconstructs data but struggles with trade-offs between fragment size and privacy leakage. PrivateTrie \cite{8509300} constructs a trie structure for iterative data collection, but its communication overhead scales still quadratically with data length. Therefore, there is a need for more efficient solutions, motivating our proposed protocol.

Our protocols relies on the stability-based histogram technique \cite{korolova2009releasing} for collecting new data. The original stability-based histogram was proposed under the central DP model. To protect against direct data collection and tracking by the server, we implement the ESA framework. We introduce an auxiliary server that only receives encrypted messages and their hashed values to construct the histogram. This setup guarantees that the auxiliary server cannot trace the original messages. Additionally, messages passing through the shuffler are further encrypted using the auxiliary server's public key. Our protocol avoids segmenting the original messages and the computational cost for the reconstruction, thus offering an efficient algorithm that achieves accuracy comparable to the central DP model while providing a local-like privacy guarantee.

In summary, our paper makes the following contributions:

1. We propose a unified framework for LDP frequency estimation. By instantiating the corresponding parameters, our framework can be used to analyze many protocols that falls in the``pure'' Local Differential Privacy framework, as well as Apple's CMS.  Moreover, our design allows the length of the privatized message to be a tunable parameter, which offers flexibility for different communication considerations.
By choosing the appropriate parameter, 
this approach significantly reduces communication costs while maintaining the same level of privacy protection.

2. We provide a general utility guarantee of our framework, beyond which, we propose a parameter optimization algorithm, tailored to specific item frequency regimes.  {In addition, we present a correction to the variance expression in Apple's original paper \cite{apple2017learning}}. Our results also suggest that most of the LDP frequency estimation protocols, such as Apple's CMS, choose suboptimal parameters for the randomization step when there is a target frequency regime.

3. We introduce a protocol collecting data within an unknown domain, that leverages the stability-based histogram technique combined with the Encryption-Shuffling-Analysis (ESA) framework. By incorporating an auxiliary server that constructs the histogram without accessing the original data messages, this protocol achieves accuracy comparable to the central DP model while providing local-like privacy guarantees. It also avoids the segmenting and reconstruction steps in all existing work under LDP, thus significantly reducing the computational cost.

4. We visualize our theoretical analysis and conduct extensive experiments with real data. Our results suggest that: 1. GCMS constantly provides better utility than CMS under the same communication cost and under the same privacy leakage. 2. With optimized mechanism parameters, our OCMS achieves lower variance for specific target frequencies, enabling clients to obtain more precise results for query items within certain frequency ranges.

\section{Preliminaries}
In this section, we introduce several techniques relevant to this paper.
\subsection{Privacy Models}

Differential privacy \cite{DMNS06} is a mathematical guarantee about database privacy.

\begin{defi}[$(\epsilon,\delta)$-DP \cite{DMNS06}]\label{DP:def}
	A randomized algorithm $\mathcal{M}: \mathcal{X} \rightarrow \mathcal{R}$ is $(\epsilon,\delta)$-differentially private if for every pair of neighboring datasets $X,X' \in \mathcal{X}$, and for every event $S \subseteq \mathcal{R}$,
 \begin{equation}\label{eq:DP-def}
  \operatorname{Pr}(\mathcal{M}(X) \in S) \leq e^\epsilon \cdot \operatorname{Pr}(\mathcal{M}(X') \in S)+\delta.
 \end{equation}
 When $\delta$ = 0, we say it is $\epsilon$-DP (or pure DP).
\end{defi}

In central DP model, there is a trusted server collecting raw data from clients and applies a randomized algorithm to the aggregated dataset to produce differentially private outputs. We define the neighboring datasets as the datasets that arbitrarily differ in the values at most one entry. Thus, \eqref{DP:def} guarantees that the presence or absence of an individual's data in a dataset does not significantly affect the outcome of queries, controlled by the privacy parameters $\epsilon$ and $\delta$.

Local Differential Privacy (LDP) \cite{Kasiviswanathan2008WhatCW} removes the need for a trusted servers by ensuring that the privacy of each individual's data is protected before it is collected by the server. In LDP, each client applies a randomized algorithm to their own data locally, and only the perturbed data is sent to the server. Here, we define the neighboring datasets as any pair of values from the input support, and \eqref{DP:def} guarantees that any input values are indistinguishable.



Shuffle Differential Privacy (Shuffle DP) \cite{10.5555/3310435.3310586} is an intermediate model between Central DP and LDP that enhances the privacy guarantees of LDP by adding a shuffling step. In this model, each client first perturbs their data locally (as in LDP) and then submits it to a shuffler, which randomly permutes the data before sending it to the server. 
Feldman et al. \cite{feldman2023stronger} have shown that the random permutation has a strong privacy amplification effect. 

\begin{thm}[Privacy amplification by shuffling \cite{feldman2023stronger}]\label{thm.shuffle}
For local randomizer with $$\epsilon\le \log\left(\frac{n}{8\log(2/\delta)}-1\right),$$for any $\delta\in[0,1]$, shuffling can achieve $(\epsilon_c,\delta)$-DP with 
    \begin{equation}\label{eq:sfDP}
        \epsilon_c\le \log \left(1+(e^{\epsilon}-1)\left(\frac{4\sqrt{2\log(4/\delta)}}{\sqrt{(e^{\epsilon}+1)n}}+\frac{4}{n}\right)\right),
    \end{equation}
where $n$ denotes the number of data items. 
\end{thm}

The Encryption-Shuffling-Analysis (ESA) framework builds on the principles of Shuffle DP to further enhance security: clients encrypt their output with the server's public key. Therefore, only the server can decrypt and obtain the plaintext. This model has been widely adopted by Apple\cite{apple2017learning}, Google\cite{Rappor}, and Microsoft\cite{10.5555/3294996.3295115}.



\subsection{Pure LDP Protocol}

The LDP mechanisms designed specifically for frequency estimation are called Frequency Oracles. Oracles vary in their construction, accuracy guarantees, and the size of the domain for which they are best suited. In terms of the design of the frequency oracle, the concept of a ``pure" protocol is introduced in \cite{203872}. A pure protocol requires the probability that any value $d_1$ is mapped to its own support set to be the same for all values. We use $p$ to denote this probability. It also requires a value $d_2\neq d_1$ to be mapped to $d_1$’s support set with probability $q$. This must be the same for all pairs of $d_1, d_2$. 
The pure protocol has the following two major advantages:

\begin{itemize}
    \item The pure protocol simplifies the privacy analysis for LDP. A pure protocol satisfies $\epsilon$-LDP as long as 
$${p}/{q} \le e^{\epsilon}.$$

    \item The pure protocol provides a closed-form estimator for frequency estimation. Let $v_i$ be the privatized data submitted by client $i$. The unbiased estimator for number of times that $d$ occurs is
    \begin{equation*}
        \hat{f}(d) = \frac{\sum_{i=1}^N \mathbbm{1}_{\{d\in support(v_i)\}} - nq}{(p-q)^2}.
    \end{equation*}
    The corresponding variance is
    \begin{equation*}
        \text{Var}[\hat{f}(d)] = \frac{f(d)p(1-p) + (n-f(d))q(1-q)}{(p-q)^2}.
    \end{equation*}
\end{itemize}

\section{Related Works}
In this section, we present related works to this paper, including LDP mechanisms and real-world implementations for frequency estimation, as well as works that provide functions for data collection with unknown domain.


\subsection{LDP Mechanisms for Frequency Estimation}


{Many existing protocols can be viewed as special cases of a pure protocol.}
Let us start with direct perturbation for a basic binary model. Binary Randomized Response (BRR), where $p = \frac{e^{\epsilon}}{e^{\epsilon} + 1}$ and $q = \frac{1}{e^{\epsilon} + 1}$, is proven to be optimal \cite{Extreme_ldp} for binary data collection. General Randomized Response (GRR), also known as Direct Encoding (DE) \cite{203872}, where $p = \frac{e^{\epsilon}}{e^{\epsilon} + |\mathcal{D}| - 1}$ and $q = \frac{1}{e^{\epsilon} + |\mathcal{D}| - 1}$, can handle high-dimensional data for $d \ge 2$, but its utility decreases significantly for large $d$.


Unary encoding (UE) reduces the sensitivity of high-dimensional data by converting a value $d \in \mathcal{D}$ into a binary vector of size $|\mathcal{D}|$, where only the $d$-th bit is $1$, and the rest are $0$s. Symmetric Unary Encoding (SUE) \cite{203872} perturbs each bit of the unary encoded vector independently using $\epsilon/2$ BRR, ensuring a symmetric property $p + q = 1$. However, this symmetry does not minimize variance. Optimized Unary Encoding (OUE) \cite{203872} improves upon SUE by minimizing variance, specially under small $\epsilon$ and low frequencies. While UE-based methods enhance utility for direct encoding, they require each client to send at least $|\mathcal{D}|$ data to the server, increasing communication overhead.


To reduce communication cost,
transformation-based approaches, such as Hadamard Randomized Response (HRR) \cite{apple2017learning} and S-Hist \cite{bassily2015local}, compress high-dimensional data into a smaller domain. Subset Selection (SS) \cite{wang2016mutual} reduces communication by randomly selecting $s$ items from the domain $\mathcal{D}$. Among these, the hashing-based methods are most widely studied.
Hash encoding maps high-dimensional data to a smaller domain \cite{203872}, then uses either direct perturbation or unary encoding to release the hashed data.

Real-world LDP implementations, such as Google's RAPPOR \cite{Rappor}, O-RAPPOR \cite{10.5555/3045390.3045647}, and Apple's Count Mean Sketch \cite{apple2017learning}, usually combine hash-based methods with UE. That is, the data is first hashed to a smaller domain, then released with UE. 
RAPPOR uses cohorts and LASSO regression to reduce information loss from hash collisions, while Apple's Count Mean Sketch uses a probabilistic data structure to average aggregated hashed items. 
With a large number of hash functions, the summation of the counts for the specific hashed values obtained by different hash functions converges to the true count plus a calculable bias.
However, the parameters used in this protocl are suboptimal, both in terms of accuracy and communication cost, as we will show later in this paper.

{There are a line of work focusing on improving the three way tradeoff between accuracy, efficiency, and privacy of the LDP-based frequency estimation frameworks. Projective Geometry Response\cite{feldman2022private} uses geometric transformations for efficient, accurate estimations, while Adaptive Online Bayesian Estimation\cite{aydin2024adaptive} applies Bayesian updates to adjust estimations in real-time, though both can be computationally complex. Wiener Filter-Based Deconvolution  \cite{10179389} enhances noise reduction using filter-based techniques, at the cost of higher computational overhead. Notably, PK-RAPPOR \cite{9826894} incorporates prior knowledge of item frequency rankings, improving accuracy specifically for frequently occurring items. However, obtaining accurate item frequency rankings is challenging in practice, and these rankings can change over time. Nonetheless, the concept of leveraging prior information for improved mechanism design remains a valid and valuable approach, which also motivates our design.}

\subsection{Data Collection with Unknown Domain}

Most of the existing works on LDP for frequency estimation requires that the server knows the data domain $\mathcal{D}$. However, there are cases where the server does not fully know the input dictionary, especially for lengthy strings like URLs or full names. Few works address data collection with an unknown domain. Fanti et al. \cite{45382} propose an algorithm for RAPPOR that splits a message into disjoint segments, releases each segment with RAPPOR, and reconstructs the message using expectation maximization (EM). This process is highly costly due to the exhaustive search needed for EM. A private sequence fragment puzzle method based on CMS \cite{apple2017learning} also uses a segmentation-release-reconstruct approach. This method faces similar issues: small fragment sizes lead to numerous fragments and high reconstruction costs, while large fragment sizes increase privacy leakage.
PrivateTrie \cite{8509300} constructs a trie data structure to store string values iteratively from each client, reducing the reconstruction from a graph to a tree. However, it comes with quadratic communication cost relative to the data length.

\section{Generalized Count Mean Sketch}\label{Sec:OCMS}

We first study the problem of privacy-preserving frequency estimation of clients' private data. We consider a set of $n$ clients, each possessing private data $d_i \in \mathcal{D}$, where $i$ denotes the client index. Each client aims to collaborate with a server to receive some service but does not fully trust the server and will only contribute their data if it is properly privatized and secured. To achieve this, each client submits a privatized and secured version $v_i$ of $d_i$ to the server. The server, requires the following functionalities:
\begin{itemize}
    \item Upon receiving the set $\{v_i\}_{i=1}^n$, the server aims to estimate the number of clients who possess a specific data $d$, denoted as $f(d)$. 
    \item  The server may have complete, partial, or no knowledge of the domain $\mathcal{D}$. Therefore, the server requires some of the popular data strings collected to keep itself up-to-dated.
\end{itemize}

\textbf{Utility Definition:} When performing frequency estimation, the servers hopes that the estimator $\hat{f}(d)(\{v_i\}_{i=1}^n)$ 
has a small mean square error (MSE) with respect to the true frequency $f(d)$:
$$E[(\hat{f}(d) - f(d))^2].$$

Local mechanisms feature lightweight designs, and communication cost and computation cost should be considered as part of the system requirements. We define the communication cost as the number of bits needed to transmit the plaintext message of each client. We note that the actual bits transmitted end-to-end can increase after encryption. However, the length of the plaintext is a more accurate representation of the communication cost across various encryption protocols. Each client's input should not burden the server during aggregation. To this end, the framework should feature low computation cost, which is defined as the computation time at the server per client's message.

Given these considerations,
we propose the Generalized Count Mean Sketch (GCMS) protocol. The term ``generalized" refers to the increased flexibility in designing parameters such as $p$ and $q$, compared to Apple's CMS. Building on this, we propose our optimal Count Mean Sketch (OCMS) that optimizes $p$ and $q$ for specific frequency regimes. We present our protocol within the ESA framework.

\subsection{General Count Mean Sketch (GCMS) Framework}

  \textbf{Preparation:} 
  Before data collection, the server generates $k$ independent hash functions $\mathcal{H} \overset{\Delta}{=} \{h_1,h_2,...,h_k\}$, where each hash function deterministically maps any input to a discrete number in $ [\mathbf{m}]=[1,...,m]$, with $m$ being the hashing range, $\mathcal{H}$ being the universe of hash functions. The server then sends $\mathcal{H}$ and its public key $pk$ to each client.
  
The entire data collection pipeline is illustrated in Figure 1. It consists of three distinct phases, each operating on different platforms:
\begin{enumerate}
    \item The initial phase involves all on-device algorithms, including hash encoding and privatization, and encryption with the server's public key $pk$, as described in Algorithm 1.
    \item The encrypted privatized data is then transmitted through an end-to-end encrypted channel and received by the shuffler. The shuffler then forwards the data to the server after a random shuffling operation.
    \item Finally, the server decrypts the input data, performs data aggregation, and obtains the Frequency Lookup Table - a sketch matrix $\mathcal{M}$, via Algorithm 2.
\end{enumerate}

\begin{figure*}[t]
    \centering
    \includegraphics[width=0.55\textwidth]{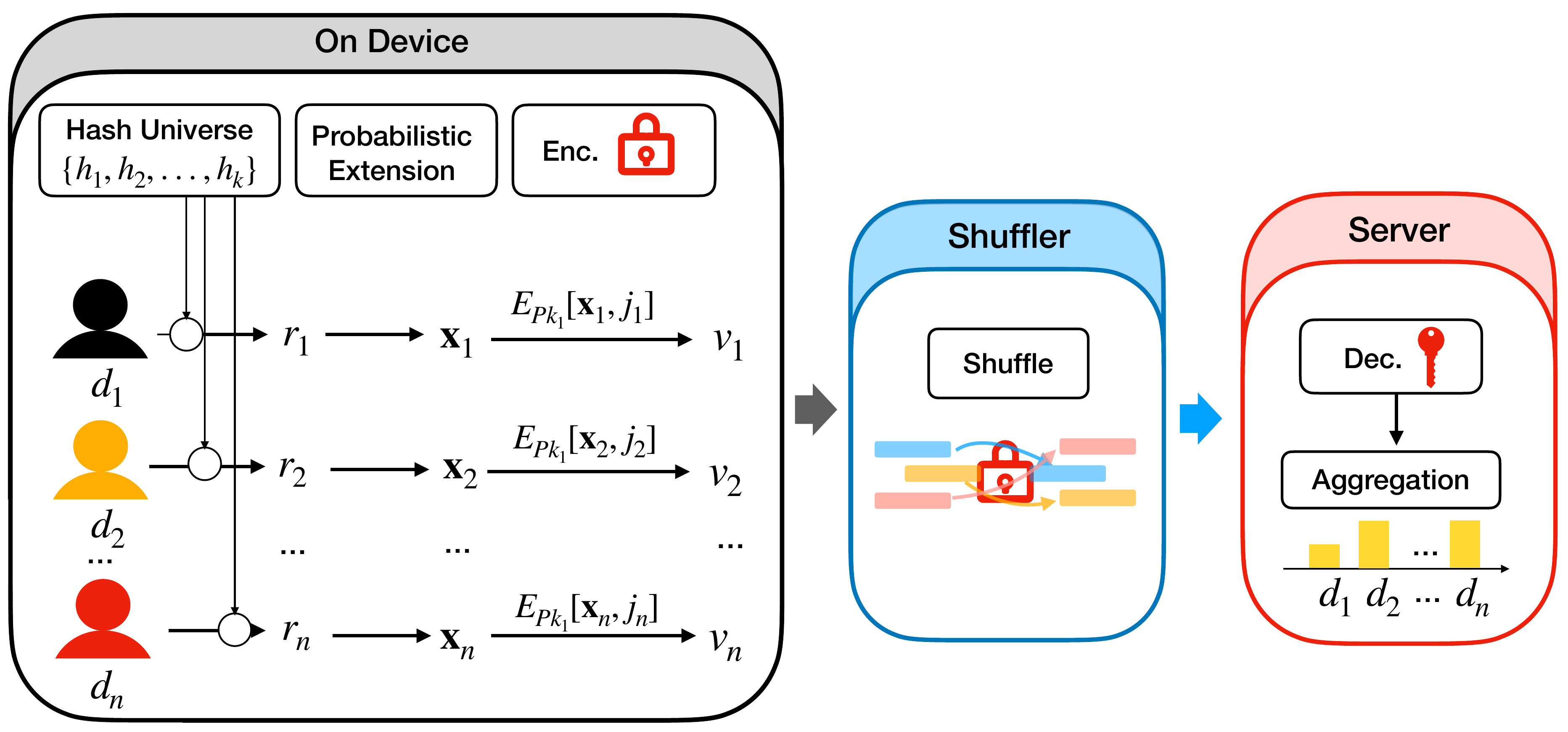}
    \caption{Illustration of the optimal Count Mean Sketch with the ESA framework.}
    \label{fig:OCMS}
\end{figure*}

\begin{algorithm}[t]
\caption{On-device LDP Algorithm}\label{alg:LDP_on_device}
 \hspace*{\algorithmicindent} 
 \textbf{Input:} $\mathcal{H}$: the hash universe;
$[\mathbf{m}]$: extension domain;
$d$: raw message; $s$: message size; $p$: inclusion probability; $pk$: server's public key.\\
 \hspace*{\algorithmicindent} 
 \textbf{Output:} Encrypted privatized message $v$.
\begin{algorithmic}
\item Randomly select $h_j\sim \mathcal{H}$; 
\item Calculate the hashed value $r = h_j[d]$;
\item Initiate output vector $\mathbf{x}$ as an empty set;
\item Add $r$ to $\mathbf{x}$ with probability of $p$;
\If {$r$ is added to $\mathbf{x}$}
\State Randomly select $s-1$ elements from $[\mathbf{m}]/r$;
\Else
\State Randomly select $s$ elements from $[\mathbf{m}]/r$;
\EndIf
\item Add selected elements to $\bold{x}$;
\item Encrypt with Server's public key
$$v = E_{pk}[\textbf{x}, j];$$
\State \Return $v$
\end{algorithmic}
\end{algorithm}

\begin{algorithm}[t]
\caption{Constructing Sketch Matrix }\label{alg:LDP_server1}
 \hspace*{\algorithmicindent} 
 \textbf{Input:} $m$: domain of hash function, $k$: size of the hash universe, $\langle\mathbf{x}_i, j_i\rangle_{i=1}^n$: decrypted clients' messages.\\
 \hspace*{\algorithmicindent} 
 \textbf{Output:} Sketch matrix $\mathcal{M}$.
\begin{algorithmic}
\item Initialized $\mathcal{M} = [0]^{[m \times k]}$;
\For{Each pair of $\langle\bold{x}, j\rangle$}
\For{Each value $x\in \bold{x}$}
\State $\mathcal{M}_j[x] \gets \mathcal{M}_j[x] + 1$;
\EndFor
\EndFor
\State \Return $\mathcal{M}$
\end{algorithmic}
\end{algorithm}

\begin{algorithm}
\caption{Frequency Estimation }\label{alg:LDP_server2}
 \hspace*{\algorithmicindent} 
 \textbf{Input:} $d$: item to check frequency, $p$: inclusion probability, $s$: message size, $m$: domain of hash function, $\mathcal{H}$: hash universe, $\mathcal{M}$: aggregated sketch matrix.\\
 \hspace*{\algorithmicindent} 
 \textbf{Output:} Estimated frequency $\hat{f}(d)$.
\begin{algorithmic}
\item Calculate $q$ according to:$$q = \frac{p(s - 1) + (1 - p)s}{m - 1};$$
\item Initialized count $C(d) = 0$;
\For{$j \in [1,2,...,k]$}
\State $C(d) \gets C(d) + \mathcal{M}_j[h_j[d]]$;
\EndFor
\item Get $\hat{f}(d)$ according to:
$$\hat{f}(d) = \frac{C(d)-\frac{pn}{m}-qn\left(1 -\frac{1}{m}\right)}{(p-q)\left(1 -\frac{1}{m}\right)};$$
\State \Return $\hat{f}(d)$
\end{algorithmic}
\end{algorithm}

  
  

We now detail operations in each phase.

\textbf{Phase 1: On-device operations.}
We present the process for a single data privatization and encryption in Algorithm 1, which consists of the following steps.

  
{Hash Encoding:} 
  Each client first uniformly selects a hash function from $\{h_1, h_2, \ldots, h_k\}$ 
  and calculates a hashed value $r=h_j(d)$
 of their raw data $d$. By construction, $r$ is an integer within $[1,m]$.
  
{Probabilistic Inclusion:}
Each client initializes their privatized vector $\textbf{x}$ as an empty set, then adds $r$ to $\textbf{x}$ with probability $p \in [0.5,1]$.

{Probabilistic Extension:}
Set the extension domain for each client as $[\mathbf{m}]/r = \{1, 2, \ldots, r-1, r+1, \ldots, m\}$.
If $r$ is added to $\textbf{x}$, then uniformly select $s-1$ elements from $[\mathbf{m}]/r$ and append them to $\textbf{x}$. If $r$ is not added to $\textbf{x}$, then uniformly select $s$ elements from $[\mathbf{m}]/r$ and append them to $\textbf{x}$. Return the privatized vector along with the selected hash index $\langle \textbf{x}, j \rangle$.

{Encryption and Release:}
Each client encrypts $\langle \textbf{x}, j \rangle$ with the server's public key $pk$ to obtain $v = E_{pk}[\textbf{x}, j]$, then releases $v$ to the shuffler.

Compared to the original CMS, our GCMS uses a novel method to constructing the privatized message (steps (b)\&(c)): the length of our privatized message is $s$, with $s\le m/2$, whereas in Apple's CMS, each client's input is a binary vector of size $m$. Our construction ensures that the probability of adding $r$ is $p$, and the probability of adding all other hash values $[\mathbf{m}]/r$ is $q$, where 
 \begin{equation}\label{eq:qm}
     q = \frac{s-p}{m - 1}.
 \end{equation}
This achieves the same randomness as in perturbing a binary vector, while significantly reduces the communication cost in terms of plain text length.


\textbf{Phase 2: Shuffler's anonymization and shuffling}. This phase, which is standard, involves 
\textit{Anonymization} and \textit{Shuffling}
after receiving each client's input. The shuffler removes all $id$ traceable information from the received message, e.g., the communication header. It then randomly shuffles the anonymized messages in batches and then releases the data to the server. The shuffled sequence is denoted as $\pi\left(v_1,v_2,\ldots,v_n\right)$, where $\pi$ is a random permutation function. Additionally, the shuffler may block clients who have already sent $l$ messages to limit their contributions. This helps filter out heavy clients or counter data poisoning attacks, effectively limiting sensitivity for client-level DP analysis.

\textbf{Phase 3: Server's frequency estimation.} 

The server first decrypts messages with the secret key and obtains:
$$\pi\left(\langle \mathbf{x}_1, j_1 \rangle, \langle \mathbf{x}_2, j_2 \rangle, \ldots, \langle \mathbf{x}_n, j_n \rangle \right).$$ Then, the server-side algorithm constructs a sketch matrix, in which the rows are indexed by hash functions, and each row $j$ is the sum of the privatized vector of clients who selected the hash function $h_j$. We present this step in Algorithm \ref{alg:LDP_server1}. 


To estimate the frequency of a message $d$, the server calculates the all hashed value of $d$: $\{h_j[d]\}_{j=1}^k$, and then aggregates the total count from the sketch matrix $\mathcal{M}$:
$$C(d) = \sum_{j=1}^k \mathcal{M}_j[h_j[d]].$$
An unbiased estimator for estimating the numbers of $d$ occurring is



\begin{equation}\label{eq:est}
    \hat{f}(d) = \frac{C(d)-\frac{pn}{m}-qn\left(1 -\frac{1}{m}\right)}{(p-q)\left(1 -\frac{1}{m}\right)},
\end{equation}
 where $q$ is given in \eqref{eq:qm}.

In contrast to a fixed probabilities $p,q$ in Apple's CMS,
our GCMS allows for flexibility in designing $p$ and $s$, thus enables optimal design tailored to specific item frequency regimes. We will discuss this in detail in Section \ref{sec:optimal_p}.


\begin{rmk}
    The communication cost, measured by the size of plaintext, is $\mathcal{O}(s\log m)$: there are $s$ values to transfer, and each value requires $\log m$ bits. Comparably, We note that the communication cost for Apple's CMS is $\mathcal{O}(m)$ for each client. The computation cost at the server for aggregation is $\mathcal{O}(ns)$ for our GCMS, implying that the matrix $\mathcal{M}$ needs to be updated for $n$ times, each contains $s$ operations. In contrast, Apple's CMS requires a computation cost of $\mathcal{O}(nm)$, as each bit contained in the message needs to be de-biased and then aggregated.
\end{rmk}

\subsection{Theoretical Analysis of GCMS}
In this section, we provide theoretical analysis of the privacy and utility of our GCMS. We first provide the LDP guarantee for each client's privatized vector.
\begin{thm}[Privacy of GCMS]
    The privatized vector $\mathbf{x}$ in GCMS is $\epsilon$-locally differentially private, with
    \begin{equation}\label{eq:epsilon}
        \epsilon \le \left|\log\left(\frac{p(m-s)}{(1-p)s}\right)\right|.
    \end{equation}
\end{thm}
Since the hashing domain $m$ is a constant, the privacy loss $\epsilon$ mainly depends on two parameters: $p$ and $s$.  For a large $p$ and small $s$, the mechanism tends to release only the true hash value, resulting in a weak guarantee. Conversely, when $p$ is small (close to $0.5$) and $s$ is large (close to $m/2$), the inclusion of the true hash value alongside other potential values becomes random, approximating a probability of 0.5. This randomness ensures that the true hash value cannot be distinguished, thus achieving a strong privacy guarantee.

We note that shuffling process can further amplify the above LDP guarantee by Theorem \ref{thm.shuffle}.


Next, we show the unbiased property of our estimator and give a general formulation for the MSE (variance) of the estimator in the following theorem.

\begin{thm}[Utility of GCMS]
    The estimator given in \eqref{eq:est} is an unbiased estimator of $f(d)$:
$$\mathbb{E}[\hat{f}(d)]=f(d).$$
    The variance of this estimator is
$$\text{Var}[\hat{f}(d)]=\frac{Var[C(d)]}{(p-q)^2(1-1/m)^2},$$ where
\begin{equation*}
\begin{aligned}
&\text{Var}[C(d)] = f(d) \left(p-\frac{p^2}{k}\right)\\
&+ (n-f(d))\left(\alpha_1-\frac{1}{k}\alpha_1^2 - \alpha_2\right) +\alpha_2\left(\sum_{d'\neq d}f(d')^2\right),
\end{aligned}
\end{equation*}
and $$\alpha_1 = \frac{p+(m-1)q}{m},~ \alpha_2 = \frac{(p-q)^2}{km}\left(1-\frac{1}{m}\right).$$
\end{thm}

Our variance formulation is general, and captures the original CMS and pure LDP protocols as special cases. 
The original CMS is with parameters $p=\frac{e^{\epsilon/2}}{e^{\epsilon/2}+1}$, $q=\frac{1}{e^{\epsilon/2}+1}$, and $k$ hash functions with a hashing domain of $[\mathbf{m}]$. By instantiating these parameters within our framework, we can get the same utility guarantee as Theorem 4.2 in \cite{apple2017learning}, with a correction to an error in the variance calculation on page 21.
For a pure LDP protocol, setting $k=1$ and $m=\infty$ within our framework yields the same utility guarantee as in \cite{wang2017locally}.

\begin{cor}[Theorem 2 in \cite{wang2017locally}]
The variance of the pure LDP protocol is 
\begin{equation}\label{eq:oue}
  \text{Var}[\hat{f}(d)]=\frac{nq(1-q)}{(p-q)^2} + \frac{f(d)(1-p-q)}{p-q}.
\end{equation}

\end{cor}

Among all unbiased estimators, the one with the least variance is desired. Our general formulation characterizes the dependence of the variance with all the parameters of the protocols, thus provides guidance to optimize the designed protocol.

\begin{figure}[t]
\centering 
\subfigure[$n = 100000$, $m = 1000$, $k = 500$.]
{\includegraphics[width=0.4\textwidth]{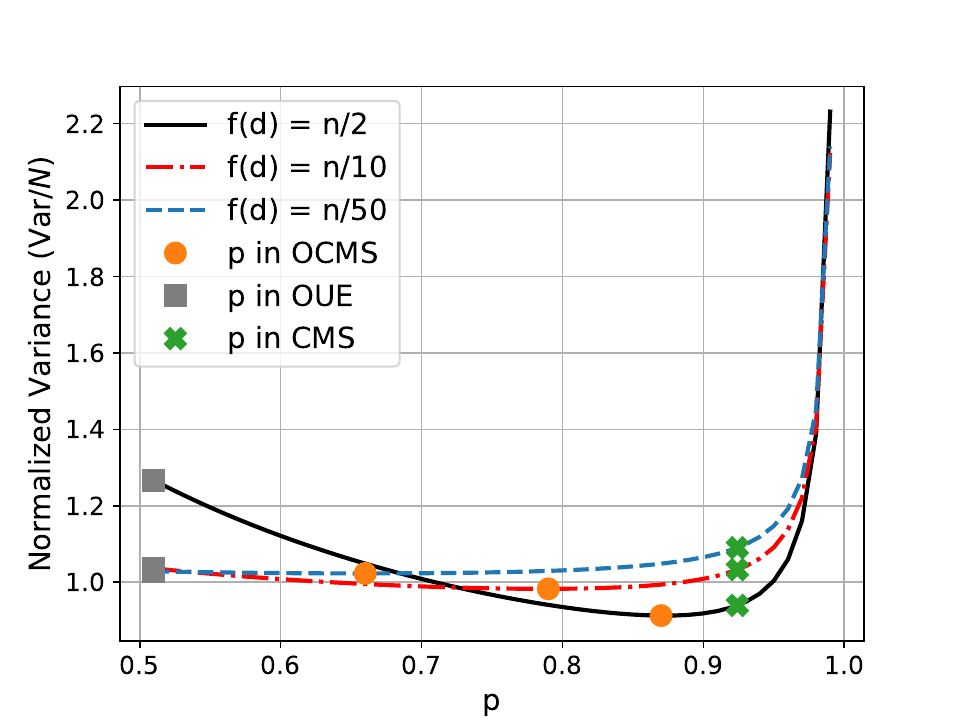}
\label{fig:OPT1}}
\subfigure[$n = 100000$, $m = 1500$, $k = 50$.]
{\includegraphics[width=0.4\textwidth]{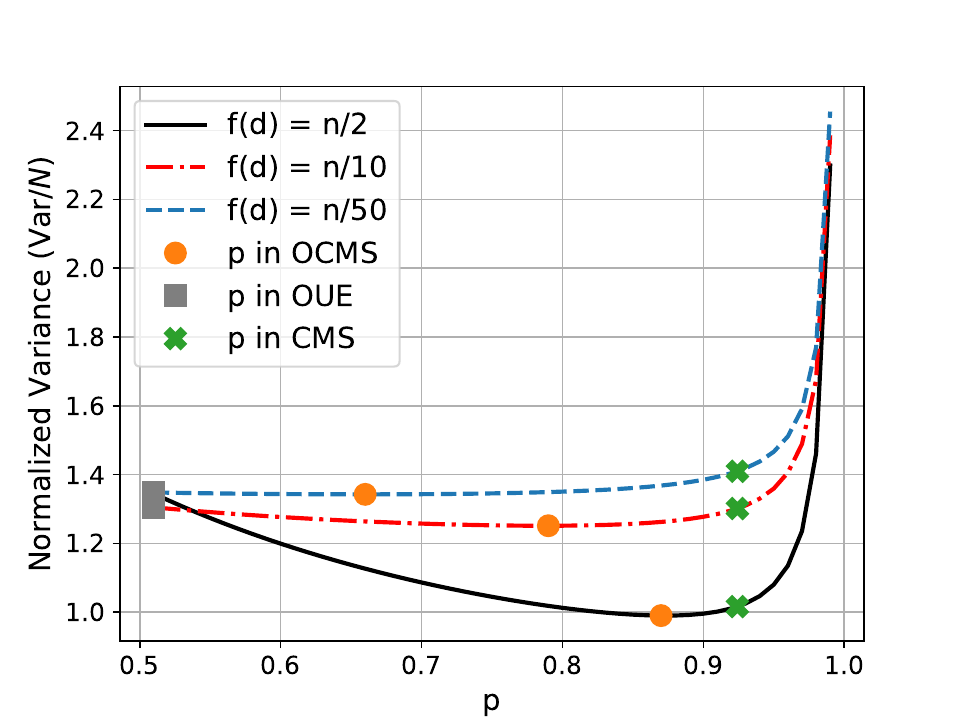}
\label{fig:OPT2}}
\caption{Choices of $p$ in OUE-LDP, Apple's CMS, and our OCMS. Optimal $p$ depends on the ratio of $n/f(d)$ and $k$, and is independent of $m$.}
\label{fig:optimal_p}
\end{figure}

\subsection{Optimal CMS (OCMS) for Targeted Frequency Estimation}\label{sec:optimal_p}

In Wang et al. \cite {wang2017locally}, Optimal Unary Encoding (OUE) was proposed to optimize the overall variance. However, OUE only considers optimizing the variance term with respect to $n$, based on the argument that most values appear infrequently and that low estimation variances for these infrequent values help avoid false positives when identifying more frequent ones. However, when analysts have a specific target frequency regime they aim to track, especially when $f(d)$ is at the same order as $n$, optimizing the entire variance including the term with respect to $f(d)$ can achieve more accurate estimation of these specific frequencies. 

In many applications, there can be some specific frequency region that is of particular interest. For example, for online shopping, the data engineers might be interested in the frequency of the most popular items, to find the top-k heavy hitter, and use them for recommendation. For emoji dashboard editing, the least popular emojis are of more interest, as to be substituted in the next version. Another example is the popularity tracking of a news or video link, where the frequency could vary from time to time.

From the formulation of the utility of the algorithm, a natural question is, can we optimize the algorithm parameters $p$ and $q$, to achieve the minimized variance, while making the algorithm achieve $\epsilon$-LDP. We note that the total client number $n$ is usually known as a priori; The variance monotonically decreases with the message size $m$ and the number of hash functions $k$. On the other hand, large $m$ increases the communication cost and server storage. Large $k$ also enlarges the server storage and the complexity of hash function design. Therefore, the selection of $m$, $k$ should be the maximum that allowed by the specific implementation requirement.

We also note that when $\epsilon$ is given, $q$ can be expressed as a function of $p$ and $m$: From \eqref{eq:qm} and \eqref{eq:epsilon}. 
\begin{equation}\label{eq:qfromp}
    q \ge \frac{m-e^{\epsilon} + e^{\epsilon}p -p}{(m-1)\left(\frac{e^{\epsilon}}{p} -e^{\epsilon} + 1\right)}.
\end{equation}

Therefore, when $n$, $m$, $k$ $\epsilon$ are considered as constant variables. The variance can be expressed as a function that determined by $f(d)$ and $p$. Then for any $f(d)$, the optimal mechanism parameter $p$ can be derived from the optimization problem in the following lemma:
\begin{lem}\label{lemma:opt_p}
    For the Optimal CMS mechanism, for a given $\epsilon$, and for a targeted frequency $f(d)$, the optimal perturbation parameter $p$ can be derived from the following optimization problem:
\begin{equation*}
    p = \text{argmin}\left\{ \frac{kw-p +\lambda{(w-1)(kw-p-wp)}}{k(1-p)^2p}\right\},
\end{equation*}
where
\begin{equation*}
    w \overset{\Delta}{=}e^{\epsilon}(1-p) + p, ~\text{and}~ \lambda = f(d)/n.
\end{equation*}
\end{lem}

\begin{rmk}
    Under a given privacy budget $\epsilon$, the optimal selection of $p$ depends on the ratio of $f(d) /n$ and $k$, and does not depend on $m$.
\end{rmk}
\begin{rmk}
The optimal parameters $p$, $q$ derived from Algorithm \ref{alg:Optimal_CMS} and Eq. \eqref{eq:qfromp} can be adapted to enhance the ``pure" protocols in \cite{wang2017locally}, by setting $k =1$ and $m = \infty$.  
\end{rmk}
In Algorithm \ref{alg:Optimal_CMS}, we present an efficient algorithm to derive the optimal $p$ under the set of parameters of $\{\lambda, \epsilon, k\}$.

It is important to note that $s$ is an integer, and the optimal $s$ derived by Algorithm \ref{alg:Optimal_CMS} involves a ceiling operation, which may result in a slightly smaller $\epsilon$ than the targeted one, yielding a stronger privacy guarantee.

\begin{algorithm}[t]
\caption{Optimal CMS parameters}\label{alg:Optimal_CMS}
 \hspace*{\algorithmicindent} 
 \textbf{Input:} $\epsilon$: privacy parameter, $\lambda$: ratio of $f(d)/n$, $k$: size of the hash universe.\\
 \hspace*{\algorithmicindent} 
 \textbf{Output:} Optimal perturbation parameter $p$, $s$.
\begin{algorithmic}
\item Create function $f(p)$ that returns:
$$ \frac{kw-p +\lambda{(w-1)(kw-p-wp)}}{k(1-p)^2p};$$
\item Initialize $p_l = 0.5$ and $p_r = 1$;
\While {$p_l\le p_r$}
\State $p_0 = p_l + \frac{p_r - p_l}{3}$;
\State $p_1 = p_r - \frac{p_r - p_l}{3}$;
\If{$f(p_0)\le f(p_1)$}
\State $p_r = p_1$;
\Else
\State $p_l = p_0$;
\EndIf
\EndWhile
\item $p\gets (p_l +p_r)/2$, $s \gets \lceil\frac{m}{1 + (1/p-1)e^{\epsilon}}\rceil$;
\State \Return $p,s$
\end{algorithmic}
\end{algorithm}

\begin{figure*}[t]
    \centering
    \includegraphics[width=0.55\textwidth]{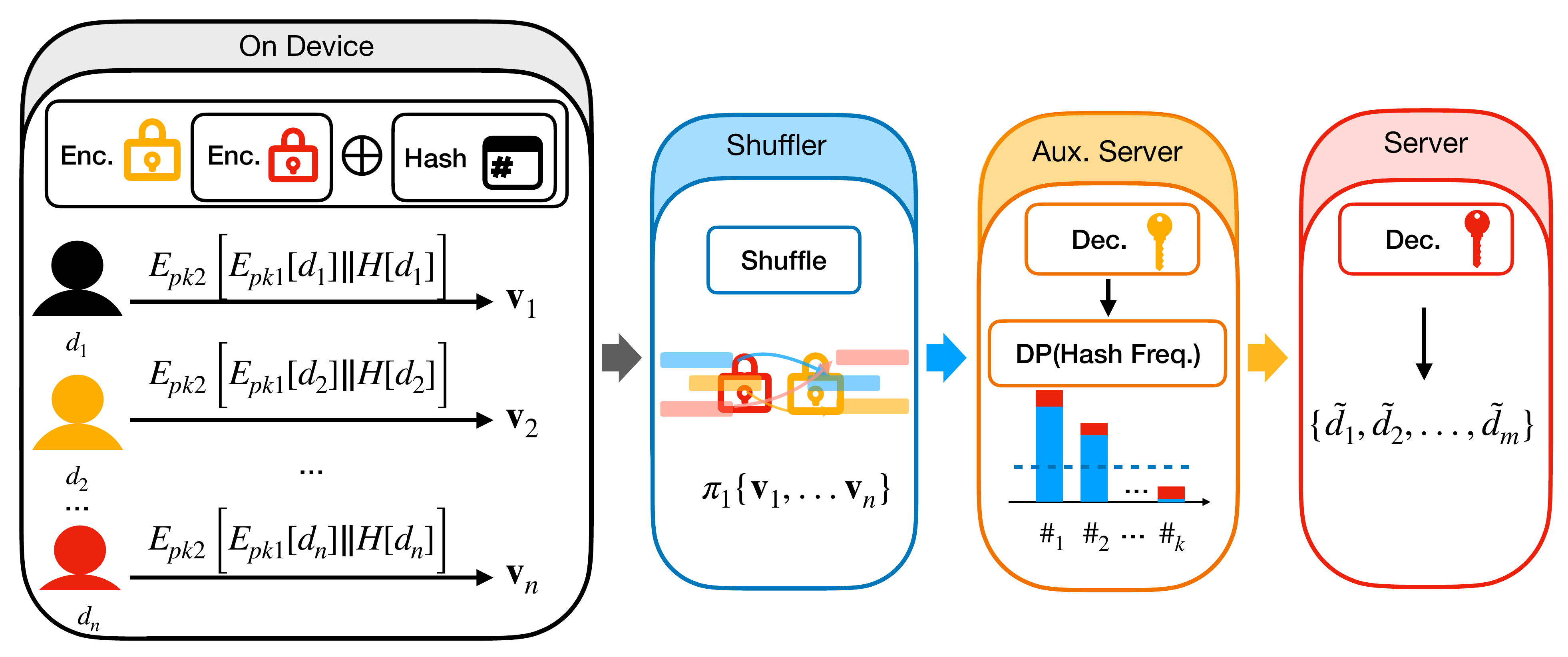}
    \caption{Illustration of the framework of privacy-preserving data collection with unknown domain.}
    \label{fig:CDP_unknown}
\end{figure*}

We note that under the CMS or similar structure, such as unary encoding, RAPPOR, etc. the total privacy budget $\epsilon$ is allocated to $p$ and $q$, and the proportion is approximately the following: 
\begin{equation*}
\begin{aligned}
&p \sim \frac{e^{\epsilon_p}}{e^{\epsilon_p} + 1},~\text{and}~
q \sim \frac{1}{e^{(\epsilon-\epsilon_p)} + 1}.
\end{aligned}
\end{equation*}
The operational meaning of $p$ is the probability of including the true item in the output, and $q$ represents the probability to include a false item to provide confusion in the result. Therefore, a mechanism with a large probability of $p$ and a small probability of $q$ leads to accurate output and enhanced utility. The extent of the increase in $p$ and the decrease in $q$ is restricted by the privacy budget $\epsilon$.

In terms of the allocation of the total budget, different mechanisms have different merits. In general, there are two different allocation principles. 1. Evenly divide $\epsilon$ to $\epsilon/2$ for $p$ and $\epsilon/2$ for $q$. Therefore: $p =  \frac{e^{\epsilon/2}}{e^{\epsilon/2} + 1}$, $q=\frac{1}{e^{\epsilon/2} + 1}$.  Google's RAPPOR and Apple's CMS, etc., fall into this category. 2. grant all $\epsilon$ to $q$, and $\epsilon_p = 0$ to $p$, therefore, $p =  1/2$, $q=\frac{1}{e^{\epsilon} + 1}$. The intuition behind this selection is that the true item is only included once with $p$, and all other items all have the probability of $q$ being included. Therefore, with a limited budget, the mechanism should prioritize minimizing $q$. The Optimal Unary Encoding (OUE) falls into this category. 

However, we argue that none of these solutions is optimal for any targeted frequency under the general CMS framework. In Fig. \ref{fig:optimal_p}, we show two different cases regarding the values of $m$ and $k$, for each case, we plot normalized variances derived in Theorem 3 with
three different values of $\lambda$. Specifically, $\lambda = 1/2$, $\lambda = 1/10$, $\lambda = 1/50$, corresponding to three levels of popularity for items. Then we fix $\epsilon = 5$ and vary the value of $p$ from $0.5$ to $1$ and showcase the corresponding variance. We highlight the optimal value of $p$ corresponding to each case, along with the selection of $p$ in CMS alike mechanisms and OUE alike mechanisms.

\section{Privacy-preserving Data Collection with Unknown Domain}\label{sec:unknown}

In the previous section, we introduced the GCMS algorithm for known item frequency estimation. However, this algorithm only functions effectively when the server knows the name of the item being queried. 


{When collecting unknown data strings, the server wants to recover and reconstruct as many data strings as possible as long as the privacy guarantee is not violated.} Collecting data with an unknown domain in an LDP manner is intrinsically challenging due to:
\begin{itemize}
\item Encoding algorithms for preprocessing (e.g., UE or hashing) require agreement on $\mathcal{D}$ between the server and the clients.
\item Segmenting the message increases the privacy loss to $k\epsilon$, where $k$ denotes the number of pieces.
    

\item Reconstructing the original input from segments incurs high computational costs, growing exponentially with the number of segments.
    
\end{itemize}

The key to avoid segmenting the original messages and the computational cost for the reconstruction is by encryption.
Therefore, we take a different approach by using central DP techniques and leveraging cryptographic tools and the ESA framework to achieve local-like privacy guarantees.


\subsection{The Protocol}
Our protocol uses the stability-based histogram technique \cite{korolova2009releasing} for collecting new data. To safeguard against direct data collection and tracing by the server, we integrate the ESA framework. Beyond the ESA framework, we introduce an auxiliary server to construct the histogram, akin to a
central curator under the central DP model, but with a crucial distinction: it does not access the original data messages. Instead, the auxiliary server only receives an encrypted version of the original message $d$ by the server's public key, $E_{pk1}[d]$, along with its hashed value $H[d]$. This ensures the auxiliary server gains no knowledge of the original messages {but the hash, which is irreversible}. To construct the histogram
the auxiliary server can count the number of message $d$ occurring by counting the hashed value $H[d]$, and each bin is represented by {a sampled } $E_{pk1}[d]$ {(two encrypted $d$ with the same $pk1$ can different)}, which can only be decrypted by the server.
To add additional layer of security protection, messages passing through the shuffler between the auxiliary server and clients are further encrypted using the auxiliary server's public key:
$E_{pk2} [E_{pk1}[d] \boldsymbol{\|} H[d]]$. The overall framework contains four phases and is shown in Fig. \ref{fig:CDP_unknown}. We next detail every step below.

Before the data collection, the server and the aux. server send their public keys, $pk1$ and $pk2$ to each client. The server and the aux. server agree on a set of privacy parameters $(\epsilon,\delta)$ for the DP guarantee of the release.

\begin{algorithm}
\caption{On-device algorithm for unknown data string collection}\label{alg:CDP_on_device}
 \hspace*{\algorithmicindent} 
 \textbf{Input:} $H$: unique hash function, $pk1$: server's public key, $pk2$: aux. server's public key, $d$: raw data.\\
 \hspace*{\algorithmicindent} 
 \textbf{Output:} Encrypted message $v$
\begin{algorithmic}
\item Calculate the hashed value $H[d]$;
\item Encrypt $d$ with $pk1$ $\gets E_{pk1}[d]$;
\item Encrypt $H[d]$ and $E_{pk1}[d]$ with $pk2$:
$$v = E_{pk2} [E_{pk1}[d] \boldsymbol{\|} H[d]];$$
\State \Return $v$
\end{algorithmic}
\end{algorithm}

\begin{algorithm}
\caption{Aux. Server's DP protection}\label{alg:CDP_aux_server}
 \hspace*{\algorithmicindent} 
 \textbf{Input:} $\pi\{v_1,..,v_n\}$: shuffled encrypted messages, $T$: Threshold for DP, $b$: scale of Laplacian noise.\\
 \hspace*{\algorithmicindent} 
 \textbf{Output:} Encrypted message set. 
\begin{algorithmic}
\item Initiate hash map $F$, release set $S$;
\item Decrypt each message in $\pi\{v_1,..,v_n\}$ and obtain:\\
Arr = $\{E_{pk1}[d_1] \boldsymbol{\|} H[d_1],...,E_{pk1}[d_n] \boldsymbol{\|} H[d_n]\}$;
\For{$E_{pk1}[d], H[d]$ in Arr}
\If {$H[d]$ in $F$}
\State $F$[$H[d]$].append($E_{pk1}[d]$);
\Else
\State $F$[$H[d]$] = [$E_{pk1}[d]$];
\EndIf
\EndFor
\For{key in $F$}
\If {len($F$[key]) + \textbf{Lap}(b) $\ge T$}
\State Randomly select $E_{pk1}[d]$ from $F$[key];
\State Add $E_{pk1}[d]$ to $S$;
\EndIf
\EndFor
\State \Return $S$
\end{algorithmic}
\end{algorithm}
\noindent\textbf{Phase 1: On-device processing.}
In this phase, each client secures his data through the following process.

{Data encryption with Server's public key:} The private data is encrypted with the server's public key, the cyphertext is denoted as $E_{pk1}[d]$.

{Data hashing:} The private data is then hashed by an hash function, denoted as $H$ (unique across clients). The hash function is identical for all clients. The hashed result is denoted as $H[d]$.

{Encryption with the aux. server's public key} The encrypted data and the hashed result are then encrypted with the aux. Server's public key. The encrypted message is denoted as $v = E_{pk2} [E_{pk1}[d] \boldsymbol{\|} H[d]]$
The on-device process is described in Algorithm \ref{alg:CDP_on_device}. Then the encrypted message is passed on to the shuffler through an end to end encrypted channel.

\noindent\textbf{Phase 2: Shuffler's anonymization and shuffling}
This phase is standard, and identical to the process described in Phase 2 of GCMS, we omit the details for simplicity.

\noindent\textbf{Phase 3: Aux. server's DP protection}
In this phase, the aux. server provides DP protection for releasing the item names. This is achieved via the following process, and is detailed in Algorithm \ref{alg:CDP_aux_server}.

{Decrypt message:} The first step is to decrypt the messages received from the shuffler with the secret key. Then the aux. server observes $\{E_{pk1}[d_1] \boldsymbol{\|} H[d_1],...,E_{pk1}[d_n] \boldsymbol{\|} H[d_n]\}$.

{Hash frequency calculation:} Since each client uses the same hash function, the hashed results from different clients with identical item must be identical. The aux. server calculates the frequency of each hashed result and attaches the corresponding encrypted data to it.

{Add DP noise:}
The aux. server adds Laplacian noise with scale $b$ to frequency of each hashed result. For those hashed results with noisy frequency larger than a threshold $T$, randomly sample from its corresponding encrypted data and release to the server.

\noindent\textbf{Phase 4: Server decrypts messages}
Finally, the server decrypts messages received from the aux. server and obtains the plaintext that contains the item names.

\subsection{Privacy Analysis}

The privacy is immediate from the stability-based histogram \cite{korolova2009releasing}. 
The relationship between the privacy parameters, the threshold $T$ and the scale of the Laplacian mechanism is presented in the following Theorem. 

\begin{thm}[Privacy of Algorithm \ref{alg:CDP_aux_server}]\label{thm:4}
    The released encrypted item set $S$ is $(\epsilon,\delta)$-differentially private with
    \begin{equation}\label{eq:thDP}
        \epsilon=\max\left\{\frac{1}{b},\log\left(1+\frac{1}{2e^{(T-1)/b}-1}\right)\right\},
    \end{equation}
    and 
    \begin{equation}
        \delta=\frac{1}{2}\exp(\epsilon(1-T)).
    \end{equation}
\end{thm}
The steps for proving Theorem 4 follow directly from \cite{korolova2009releasing}.

\section{Discussion} 

\subsection{Security Analysis}
When each party in the framework is honest but curious, the security and privacy of each client’s data are inherently guaranteed by LDP and encryption. Therefore we put more focus on potential risks if we extend our analysis to the malicious setting. Typically, the adversary aims to accomplish several objectives: (1) extract additional information about the honest clients and their messages, (2) downgrade the utility of the LDP results, and (3) disrupt the protocol, causing parties to abort. 
In a synchronized network setting, if some malicious parties decide to halt the protocol, it can be easily detected and the honest parties can simply exclude those parties from the next execution. So we mainly focus on the first two goals.

\subsubsection{Analysis for GCMS/OCMS}
In the context of our frequency estimation framework, three types of entities exist: clients, the shuffler, and the server. From an attacker’s perspective, we assume all these parties could be malicious, what’s more, different types of parties could collude. 

\textbf{Malicious Clients:} For malicious clients, the worst attack they can do is similar to the data poisoning attack~\cite{279934}: the adversary can send garbage messages as LDP reports, or register a huge amount of fake clients and let them send garbage messages as LDP reports. While truncating clients' maximum contribution can counter the former, the latter type, is considered outside the scope of this paper, as it cannot be mitigated by designing a better LDP protocol. Typically, strategies such as authentication or access control can be employed to counteract such Sybil attack.

\textbf{Malicious Shuffler:} Since all the client messages are encrypted with the server's public key, even a malicious shuffler cannot break the confidentiality of the messages. However, the malicious shuffler can still deviate from the protocol. For instance, the shuffler could launch a data poisoning attack by injecting false information as LDP reports or dropping LDP reports from honest clients. Since the shuffler represents a single point of failure, these attacks are feasible. A solution is to decentralize the shuffling process. For instance, we can deploy a group of servers to run secure multiparty computation (MPC) to achieve oblivious shuffling~\cite{eskandarian2021clarion}, and each MPC server learns nothing about the shuffled message as long as a sufficient portion of the servers are honest. Alternatively, we can also put the shuffler functionality into trust execution environments (TEE)~\cite{chen2017strongly}, which provide verifiable security. To conclude, the shuffler can be forced to behave honestly by proper cryptography approaches or trusted hardware, thus it is safe to consider the shuffler as a semi-honest party.

\textbf{Malicious Server:} The malicious server has no motivation to downgrade the utility as it benefits from better LDP outcomes. Therefore, we only check if the malicious server can infer more information about the clients. However, since all messages accessed by the server are protected by LDP, it is assured that the server cannot learn additional information beyond what is allowed by the LDP protocol.

\textbf{Collusion between parties:} 
When the shuffler colludes with the server, the worst-case scenario is that the shuffler fails to shuffle the messages, thereby eliminating the privacy amplification effect. However, the results of the protocol are still protected by LDP, although at a reduced level of privacy. In cases where the server colludes with a group of malicious clients, the server still gets no extra information of the rest of the clients due to the LDP guarantee. When the shuffler colludes with malicious clients, the worst-case situation is that the shuffler consistently accepts inputs from the malicious clients while refusing to serve the honest ones. This issue cannot be addressed through LDP design alone. It can be mitigated by implementing measures such as enforced scheduling or adopting a decentralized shuffling solution, such as MPC.

\subsubsection{Analysis for the aux. server}

Our unknown item discovery framework introduces another entity: the aux. server. Encrypting the client messages with the server's public key guarantees that the aux. server gains no information about the content of the messages. However, the aux. server can compute the hash of the specific value and compare it with the hash in each client message. In this way, the aux. server can check whether the client message is a specific value. Therefore, we require the message space to be large enough so that the aux. server cannot brute force the whole space. Another aspect is that the aux. server can also purposely decrease the frequency of a specific value. To do so, the aux. server simply computes the hash of the value and drops all messages with the same hash. Finally, the aux. server can also obtain a list of frequencies, but without knowing which frequency corresponds to which item. Because the messages are shuffled by the shuffler, the aux. server has no information regarding who sent which messages.

To conclude, the malicious aux. server could significantly change the result of the protocol, therefore, we suggest deploying the aux. server with an trusted authority, or implementing the aux. server in trust execution environments (TEE)~\cite{chen2017strongly}, which enforces the protocol execution to eliminate the extra information leakage, and meanwhile ensures that the aux. server is not involved in any collusion.


\subsection{Remarks for Implementation}

\textbf{Choice of threshold and the privacy-utility trade-off.}
The choice of $T$ and $b$ is critical in balancing the utility-privacy trade-off. 
The privacy guarantee becomes stronger with larger $T$ and $b$, while $T$ and $b$ determine the utility from different perspectives: a larger $T$
ensures that only popular data is revealed to the server, which is not ideal if the server's goal is to collect as much unknown data as possible. Conversely, a large $b$ introduces more more randomness, reducing accuracy when the server aims to learn the most popular data strings.
Therefore, the choice of $T$ and $b$ should align with the specific application scenarios.

\textbf{Joint usage of the two frameworks.} The two data collection frameworks proposed in this paper can be used jointly in a more on-demand manner in practice. For example, the server can collect popular data strings through the unknown domain collection framework and then check their frequency with the sketch matrix in the GCMS. The server can also track the popularity of targeted data by calculating an optimal $p$ and $s$ according to OCMS with previously estimated frequency, and send $p$ and $s$ to each client as their mechanism parameters in the next iteration. 
Another use-case example is the server maintaining a list of top-$k$ popular items. The server can fine-tune $p$ and $s$ according to the frequency regime of the required popularity and send them to each client. Each time the server obtains a list of data strings from the unknown domain collection framework, the server checks the frequency and updates the list. The overall privacy guarantee is obtained by composition of the privacy of GCMS amplified by shuffler (Theorem \ref{thm.shuffle}) and the privacy of data string collection (Theorem \ref{thm:4}).


\section{Key Results Visualization and Experiments}

This section visualizes the key results and experimental evaluation for our GCMS, OCMS and related works. We start by visualizing the accuracy in terms of variance, demonstrating that our approach can achieve lower variance for items with specific frequencies. Additionally, we identify the optimal choices of $p$ under various parameter settings. Finally, we conduct tests of our approach and related works using real-world databases, providing a comprehensive comparison of their performances.

\subsection{Utility of GCMS}
We use variance as the metric to measure accuracy, which can be directly translated to the mean square error of the estimation. To provide better visualization, we slightly modify original variance to relative square root variance, which is in the same order as the absolute error or relative error, and is defined as:
\begin{equation*}
    \sqrt{\text{Var}(\hat{f}(d))}\big/f(d).
\end{equation*}
This modification does not affect the optimality of the mechanism.  There are multiple factors jointly determining the variance, such as $k$ (the number of hash function), $m$ domain (hash domain size), $f(d)$ (true frequency of the queried item), and $n$ (number of LDP reports). 
Figure~\ref{fig:var_ours_cms} shows the accuracy of our approach and Apple's CMS under different parameter settings. We fix the number of LDP reports to be $100,000$, and the communication costs of different mechanisms at the same level. Then, we focus on two target query items with true frequencies $f(d)$ of $1000$ and $10000$. The results demonstrate that our approach achieves lower variance when the privacy budget $\epsilon$ is not excessively large. Apple's CMS outperforms our approach only when the $f(d)$ of the queried item is considerably large and the privacy budget is extremely high, which rarely happens in real-world scenarios. 
\begin{figure}[t]
\centering 
\subfigure[$n = 100000$, $f(d) = 1000$.]
{\includegraphics[width=0.45\textwidth]{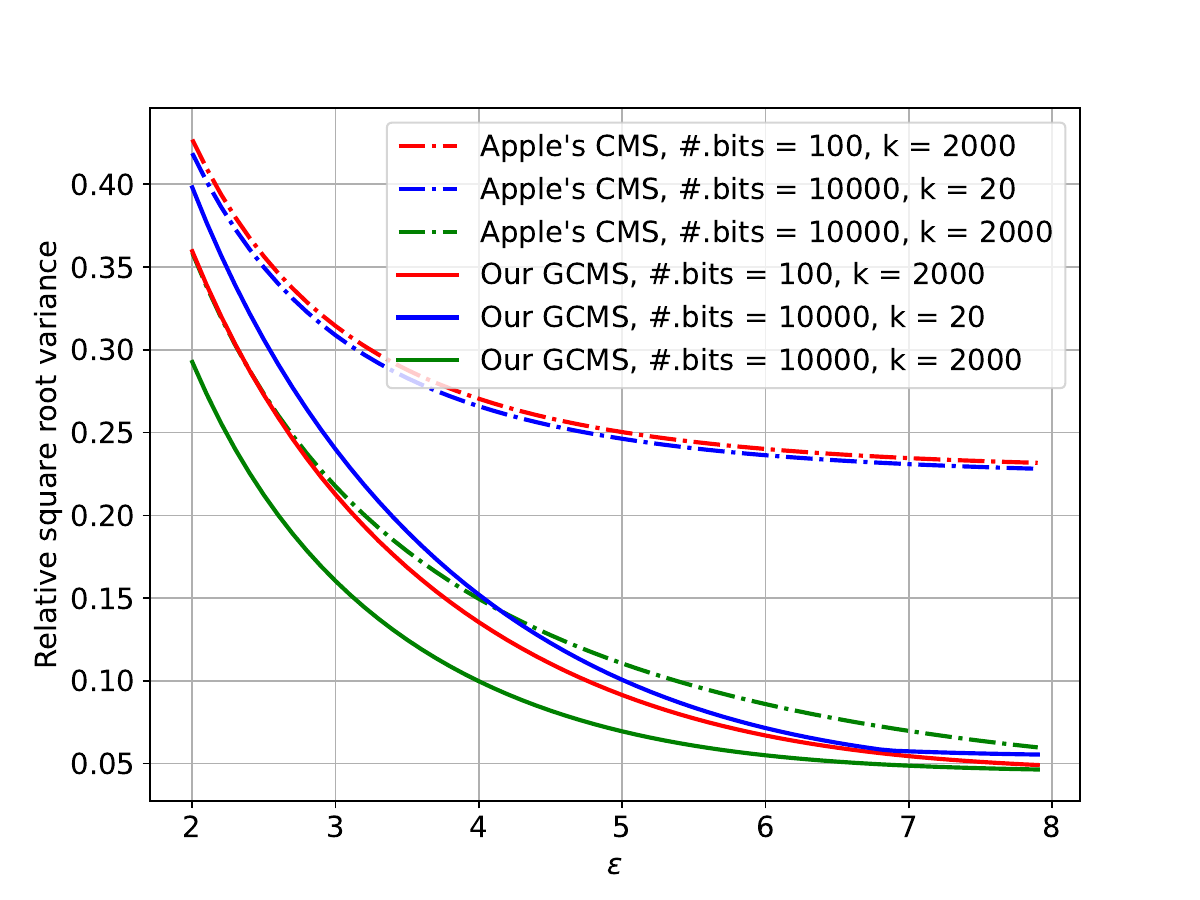}
\label{fig:var_fd_1000}}
\subfigure[$n = 100000$, $f(d) = 10000$.]
{\includegraphics[width=0.45\textwidth]{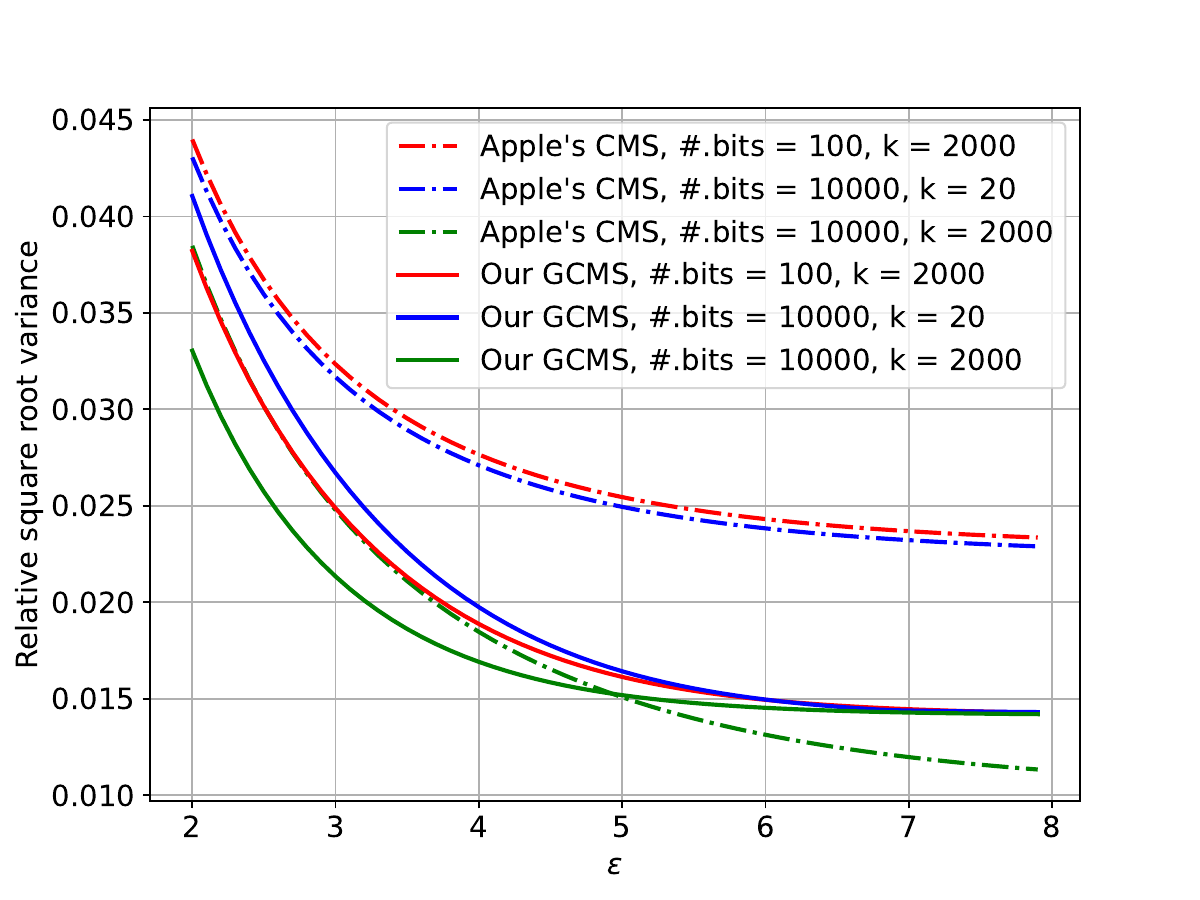}
\label{fig:var_fd_10000}}
\caption{Variance comparison between our approach and Apple's CMS under different epsilons. Lines with the same color indicate the same parameter setting. \#.bits represents the bit length of the LDP reports, which is directly related to the hash function range $m$. $k$ is the number of hash functions, $n$ is the total number of LDP reports, and $f(d)$ is the true frequency of the queried item. The perturbation parameter $p$ for our approach is $0.5$.}
\label{fig:var_ours_cms}
\vspace{-10pt}
\end{figure}




\subsection{Utility of OCMS}
Lemma~\ref{lemma:opt_p} suggests that the perturbation parameter $p$ can be adjusted for a specific frequency $f(d)$ to minimize the variance for target frequency. 
We conduct experiments to compute optimal selection of $p$ under various parameter settings, and the results are presented in Figure~\ref{fig:opt_p_experiments}. The findings reveal an interesting observation: when $\epsilon$ is small, the optimal $p$ remains $0.5$ across a wide range of $f(d)$. This is consistent with the assertions in \cite{203872} that the optimal choice of $p$ is 0.5 when the privacy budget is small. However, as the privacy budget increases, the optimal $p$ begins to diverge from $0.5$, increasing with $\epsilon$ and $f(d)$ beyond a certain threshold. Therefore, our solution can also be viewed as generalization of the optimal unary encoding approach, offering improved performance when a larger $\epsilon$ is available. Figure~\ref{fig:p_k_experiments} illustrates the relationship between the optimal $p$ and the number of hash functions $k$. generally, a larger $k$ corresponds to a higher optimal $p$. This trend indicates that our approach yields greater benefits when implemented with a more complex hashing mechanism. 


\begin{figure}[t]
\centering 
\subfigure[Optimal selection of $p$ given different $\epsilon$.]
{\includegraphics[width=0.45\textwidth]{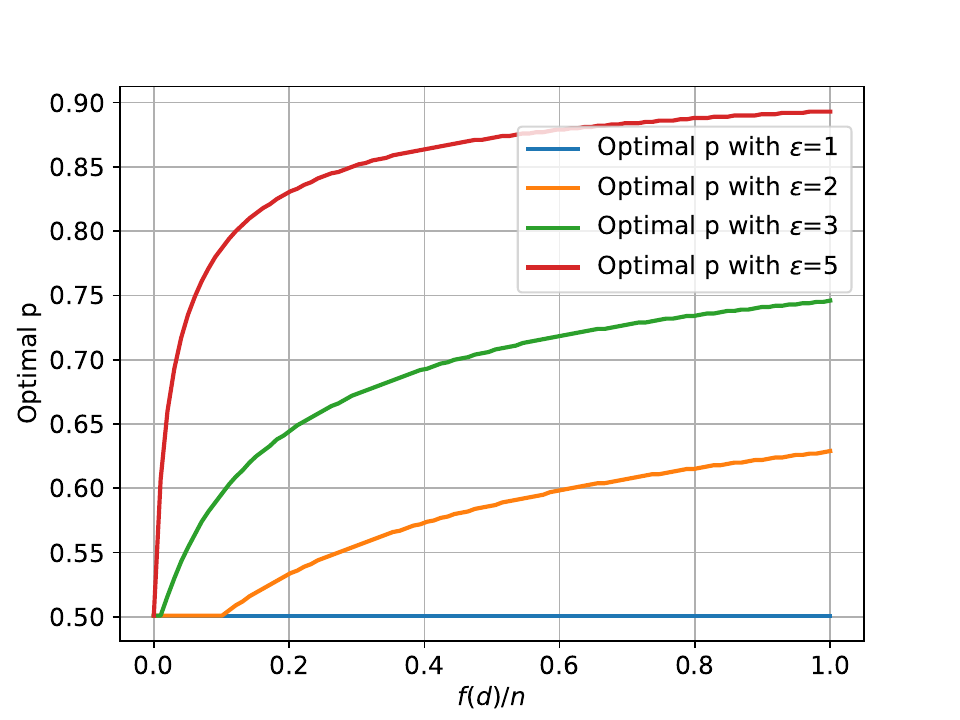}
\label{fig:p_experiments}}
\subfigure[Optimal selection of $p$ given different $k$.]
{\includegraphics[width=0.45\textwidth]{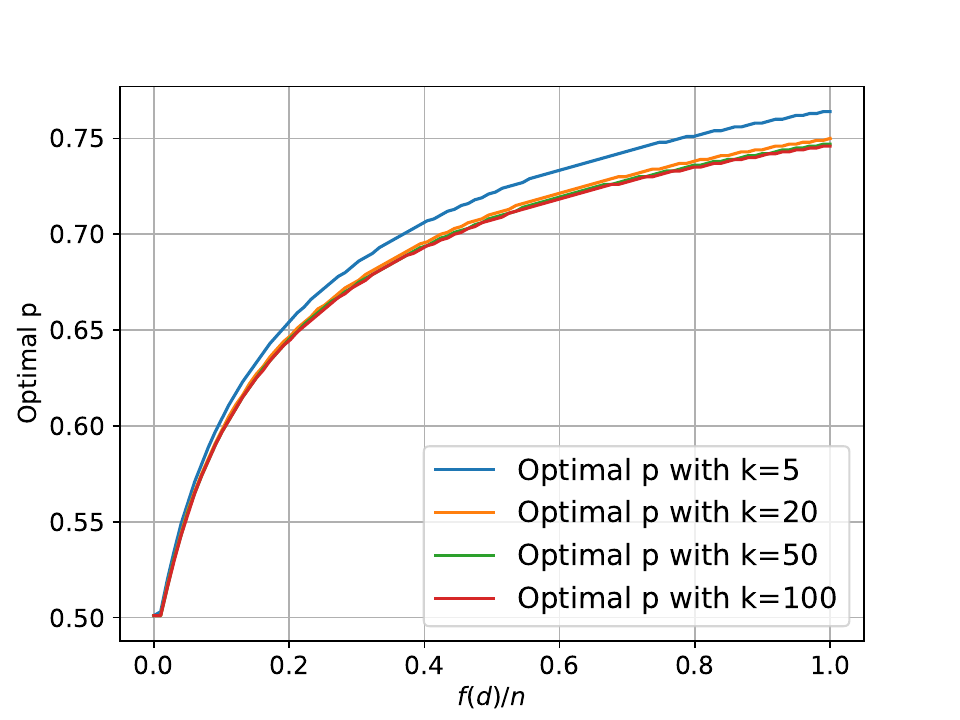}
\label{fig:p_k_experiments}}
\caption{Optimal selection of $p$ in various parameter settings. $n = 10000, k = 200, m = 200$.}
\label{fig:opt_p_experiments}
\vspace{-10pt}
\end{figure}

With the optimal selection of $p$, our approach has more advantages when we are interested in querying items within a specific frequency range. By adjusting the parameter $p$, our approach can achieve better accuracy for a target frequency, resulting in lower variance for query items with similar frequencies. Figure~\ref{fig:fd_experiments} displays the best theoretical variance achievable for each target frequency, with the variance of Apple's CMS with a fixed $p$ serving as the baseline. To ensure a fair comparison, we use the same parameter settings as in Apple's CMS experiments: $m = 1024$, $k = 65536$, $\epsilon = 4$. The results indicate that our approach significantly improves variance across all target frequencies. 


\begin{figure}[t]
    \centering
    \includegraphics[width=0.45\textwidth]{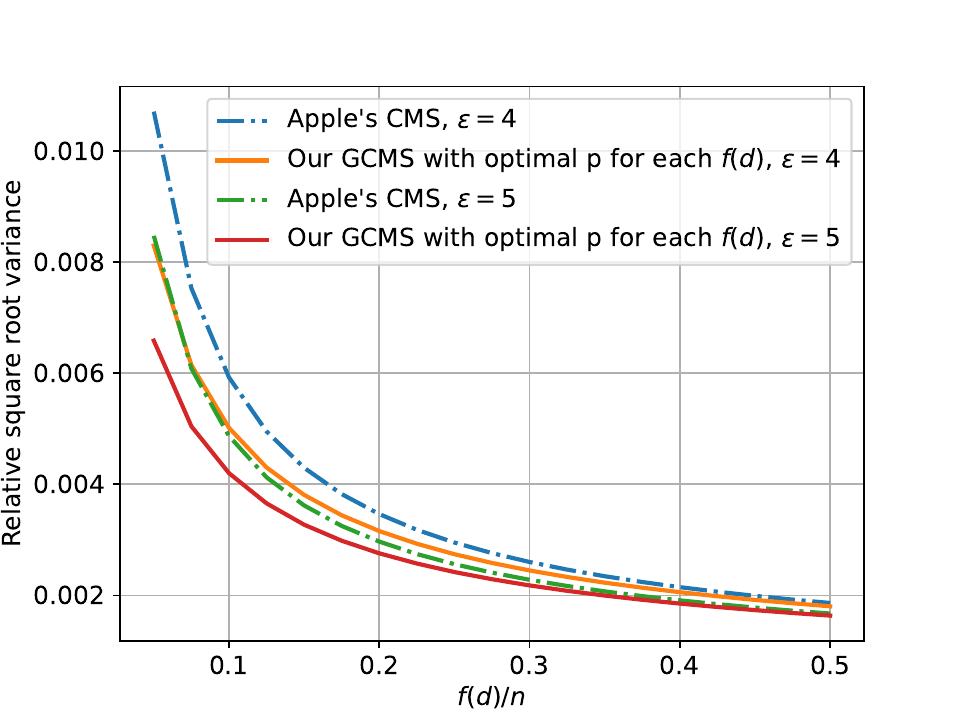}
    \caption{Variance of our approach with optimized $p$ for each target frequency $f(d)$. $n = 1000000, k = 65536, m = 1024$.}
    \label{fig:fd_experiments}
    \vspace{-10pt}
\end{figure}
\begin{figure}[t]
    \centering
    \includegraphics[width=0.45\textwidth]{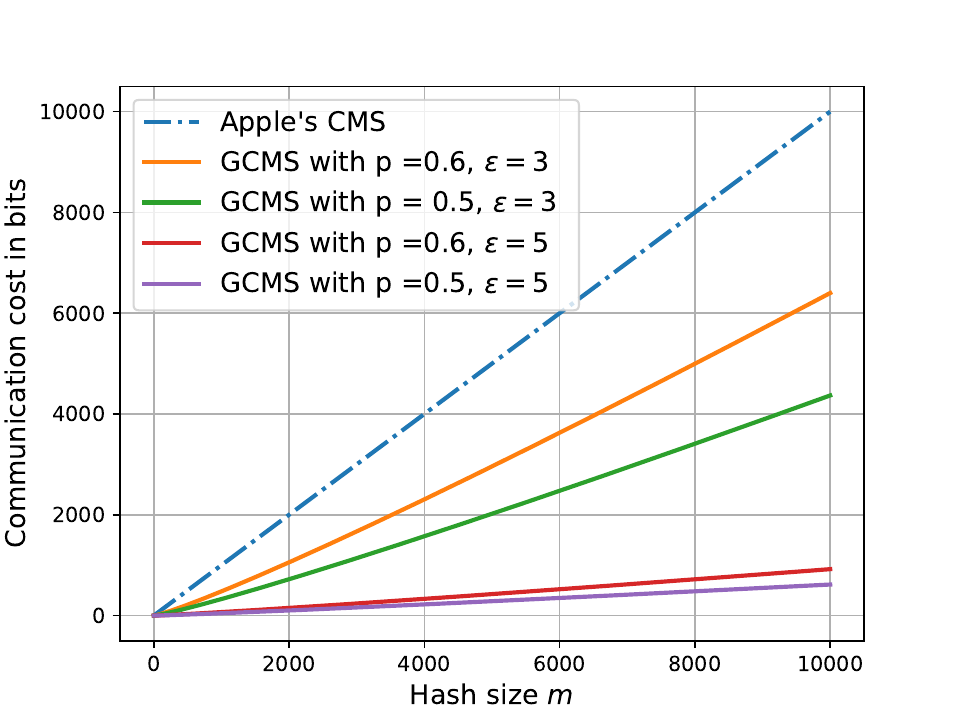}
    \caption{Communication cost comparison between GCMS and Apple's CMS under fixed $\epsilon$s and $p$s.}
    \label{fig:communication}
\vspace{-10pt}
\end{figure}
\subsection{Communication Costs}
\vspace{-5pt}
We compare the communication costs of GCMS against with that of Apple's CMS. 
As we described in Remark 1, the communication cost of the GCMS is $\mathcal{O}(s \log m)$, compared to $\mathcal{O}(m)$ for Apple's CMS. Note that while the communication cost for Apple's CMS only depends on the hash size $m$ and is independent of the mechanism parameters, the communication cost of GCMS depends on $s$, which can be viewed as a function of $\epsilon$ and $p$.

In Fig. \ref{fig:communication}, we present the communication costs when varying $p$ and $\epsilon$. We fix the privacy parameter $\epsilon$ to be $3$ and $5$, respectively, and $p$ to be $0.5$ and $0.6$, corresponding to the $p$ in OUE-LDP and in OCMS for some targeted frequency. We then vary the hash size $m$ from $0$ to $10,000$ and obtain the corresponding communication costs. We can observe that our subset-selection feature improves the communication costs significantly. Additionally, a small $p$ and large $\epsilon$ can further reduce the communication cost in GCMS.

\subsection{Experiments with Real-world Datasets}

We conduct experiments with the ``Adult Census Income" dataset\cite{misc_adult_2}, which contains census information with $48,842$ records and $15$ attributes. 
Our focus is to count the frequency of the "education" column, which contains $16$ possible values with frequencies ranging from $80$ to $16,000$. We use Python code to simulate the operations on both the client side and the server side, and ask the server to reconstruct the frequencies of all $16$ possible values. 
We target the highest frequency value of $\frac{n}{3} = 16,000$ and apply Lemma~\ref{lemma:opt_p} to determine an optimal $p$ of $0.74$. For mid-range frequency values, we pick $p = 0.57$ to achieve the best performance for items with a frequency around $1500$. Additionally, we include $p = 0.87$ as the baseline because it corresponds to the parameter setting used in Apple's CMS, which is treated as a special case in our framework. The result is depicted in Figure~\ref{fig:adults_experiments_optimized}. We observe that each specified case achieves lower variance at its target frequency, though the performance slightly decreases for non-target frequencies. Specifically, in the case of Apple's CMS, the resulting target frequency exceeds $n$ due to the large $p$. Consequently, our approach demonstrates better performance almost the entire range of true frequencies.
\begin{figure}[t]
    \centering
    \includegraphics[width=0.47\textwidth]{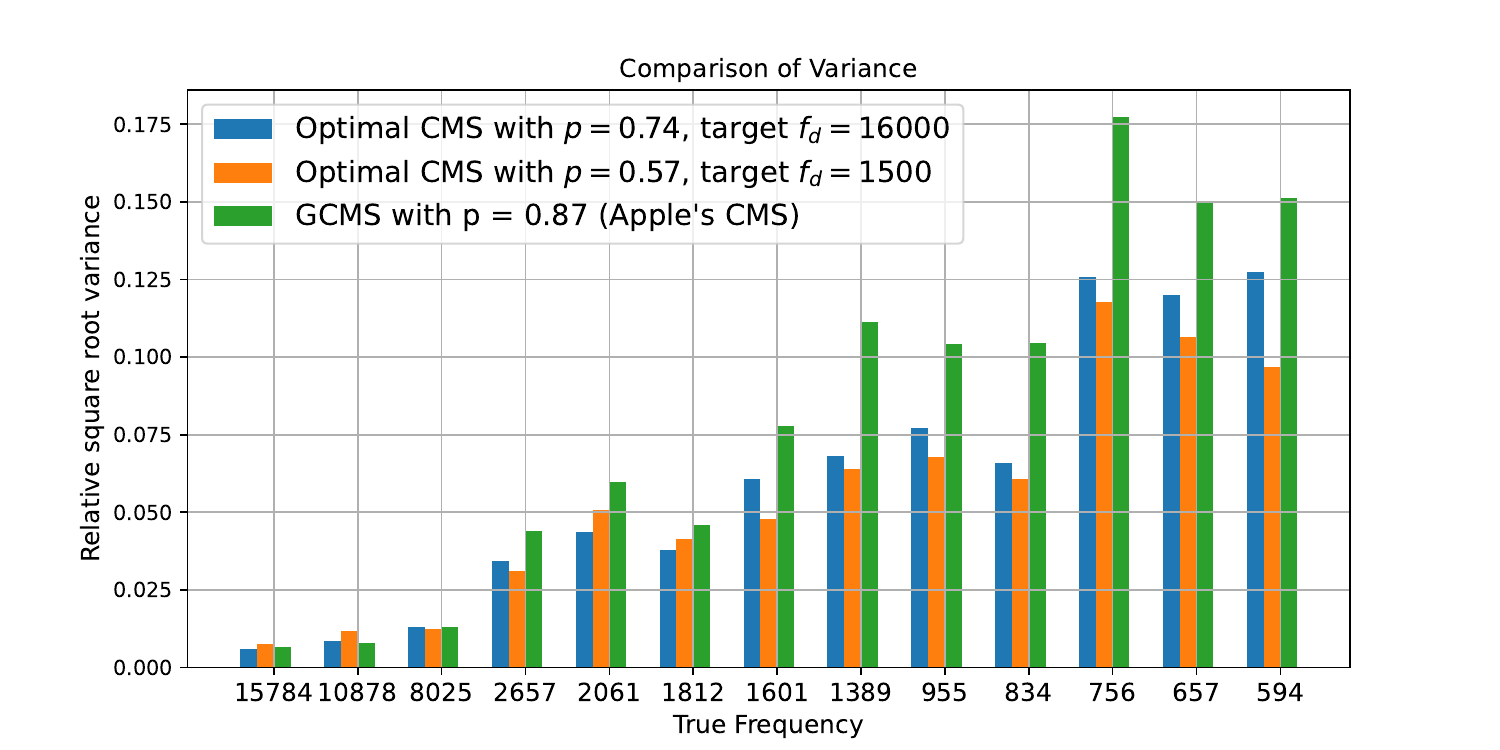}
    \caption{Variance comparison with the ``Adult Census Income" dataset with various $p$. Experiment setting: $n = 48000, m = 100, k = 100$. $\epsilon = 3.64$ for the $p = 0.74$ case, and  $\epsilon = 3.75$ for the rest two cases.}
    \label{fig:adults_experiments_optimized}
\end{figure}

\begin{figure}[t]
    \centering
    \includegraphics[width=0.38\textwidth]{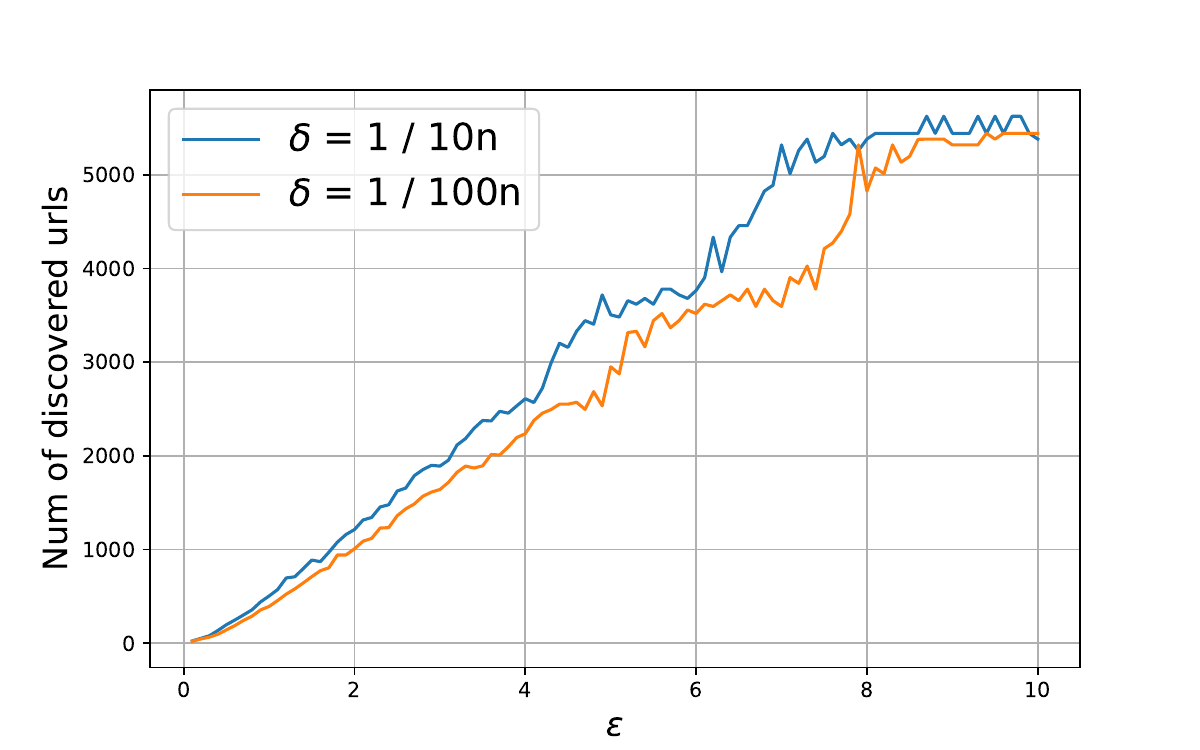}
    \caption{Relationship between privacy budget and the amount of discovered url with our approach using Kaggle's URL dataset. Experiment setting: $n = 100000$. Each experiment is repeated 30 times, and the mean value is reported.}
    \label{fig:url_experiments}
\end{figure}
\subsection{Unknown Domain Data Collection}
We conduct experiments for our unknown domain DP solution with Kaggle's URL dataset\footnote{Link: \url{https://www.kaggle.com/datasets/teseract/urldataset?resource=download}}. For convenience, we select the first $100,000$ data records and limited our analysis to the first $20$ characters of each URL. This is because many URLs contain a client-ID-like string towards the end, making them unique and thus difficult to analyze effectively for frequency-based approaches. We observe that as the privacy budget $\epsilon$ decreases, the threshold $T$ increases. Consequently, a smaller number of URLs is collected, reflecting a higher level of privacy protection at the expense of data availability. This trend is clearly illustrated in Figure~\ref{fig:url_experiments}. Additionally, as $\epsilon$ increases, the number of collected URLs tends to converge. The convergence number (approximately $5,500$ URLs) still significantly differs from the total number of distinct URLs (around $40,000$). This indicates that even at higher levels of $\epsilon$, the method does not capture all unique URLs. The underlying reason is that, according to Theorem 4, $T$ is approximately $1 - \log(2\delta)/\epsilon$, which is greater than 1 for most values of $\epsilon$. Given that the majority of URLs in the dataset have a count of 1 and the Laplacian mechanism adds zero-mean noise, most of these URLs remain unrevealed. This feature ensures that user data with outliers is intractable.

It is worth noting that our protocol offers significant improvement in accuracy. 
We also conduct experiments on the Twitter dataset, following \cite{8509300}. When compared to other existing LDP unknown domain data collection methods, PrivTrie is the only
solution that obtains a practical level of accuracy (over 0.8 F-score) under $\epsilon = 2$.
Under these same conditions, our unknown domain protocol achieves an F-score of 0.962, demonstrating a notable improvement in accuracy. However, we note that this is not a completely fair comparison, as all the methods mentioned are pure LDP protocols, whereas our approach is quasi-LDP, leveraging cryptographic techniques.

\section{Conclusion}
In this paper, we present a comprehensive framework for privacy-preserving data collection and frequency estimation, utilizing the Encryption-Shuffling-Analysis architecture. We introduce a Generalized Count Mean Sketch (GCMS) protocols that captures and optimizes various existing frequency estimation protocols. We provide its privacy and security analysis, as well as a general utility analysis that enables optimal parameter design, leading to the development of the Optimal Count Mean Sketch (OCMS) which minimizes variance for targeted frequency regimes. 
Additionally, we propose a protocol for data collection with unknown domains, addressing scenarios where the data domain is not predefined. Our approach achieves accuracy comparable to the central differential privacy (DP) model while providing local-like privacy guarantees and substantially lowering computational costs.
We also visualize our theoretical analysis and provide extensive experimental results demonstrating the effectiveness of these protocols in terms of utility-privacy trade-offs and communication cost.

\bibliographystyle{IEEEtran}
\bibliography{ref_Bo,ref_wanrong,ref_donghang,ref_jian}

\appendices

\section{Proof of Theorem 2}
Denote $r$ and $r'$ as the hashed results of two different input $d$ and $d'$. Note that $r$ and $r'$ can be resulted from different hash functions. We next examine the likelihood ratio of $r$ and $r'$ resulting the same $\mathbf{x}$. 

\noindent\textbf{Case 1:} $r\in\mathbf{x}$ and $r' \in\mathbf{x}$:

\begin{equation*}
\begin{aligned}
    \frac{\text{Pr}(X = \mathbf{x} | D = d)}{\text{Pr}(X = \mathbf{x} | D = d')} =& \frac{\text{Pr}(X = \mathbf{x} | R = r)}{\text{Pr}(X = \mathbf{x} | R = r')}\\
    =&\frac{p/\binom{m-1}{s-1}}{p/\binom{m-1}{s-1}} = 1.
\end{aligned}
\end{equation*}

\noindent\textbf{Case 2: }$r\notin\mathbf{x}$ and $r' \notin\mathbf{x}$:
\begin{equation*}
\begin{aligned}
    \frac{\text{Pr}(X = \mathbf{x} | D = d)}{\text{Pr}(X = \mathbf{x} | D = d')} =& \frac{\text{Pr}(X = \mathbf{x} | R = r)}{\text{Pr}(X = \mathbf{x} | R = r')}\\
    =&\frac{(1-p)/\binom{m-1}{s}}{(1-p)/\binom{m-1}{s}} = 1.
\end{aligned}
\end{equation*}

\noindent\textbf{Case 3: }$r\in\mathbf{x}$ and $r' \notin\mathbf{x}$:
\begin{equation*}
\begin{aligned}
    \frac{\text{Pr}(X = \mathbf{x} | D = d)}{\text{Pr}(X = \mathbf{x} | D = d')} =& \frac{\text{Pr}(X = \mathbf{x} | R = r)}{\text{Pr}(X = \mathbf{x} | R = r')}\\
    =&\frac{p\binom{m-1}{s - 1}}{(1-p)\binom{m-1}{s}}\\
    =&\frac{p}{1-p} \frac{\binom{m-1}{s}}{\binom{m-1}{s-1}}\\
    =&\frac{p}{1-p} \frac{(m-1)!/s!(m-s-1)!}{(m-1)!/(s-1)!(m-s)!}\\
    =&\frac{p}{1-p} \frac{(s-1)!(m-s)!}{s!(m-s-1)!}\\
    =&\frac{p}{1-p}\frac{m-s}{s}.
\end{aligned}
\end{equation*}

\noindent\textbf{Case 4: }$r\notin\mathbf{x}$ and $r' \in\mathbf{x}$:

Similar to case 3. 
\begin{equation*}
\begin{aligned}
    \frac{\text{Pr}(X = \mathbf{x} | D = d)}{\text{Pr}(X = \mathbf{x} | D = d')} =& \frac{\text{Pr}(X = \mathbf{x} | R = r)}{\text{Pr}(X = \mathbf{x} | R = r')}\\
    =&\frac{(1-p)\binom{m-1}{s}}{p\binom{m-1}{s-1}}\\
    =&\frac{1-p}{p}\frac{s}{m-s}.
\end{aligned}
\end{equation*}
Therefore, when $p\in[0.5,1]$ and $s\in[m/2, m]$
\begin{equation*}
    \log\left\{\frac{\text{Pr}(X = \mathbf{x} | D = d)}{\text{Pr}(X = \mathbf{x} | D = d')}\right\}\le \log\left\{\frac{p}{1-p}\frac{m-s}{s}\right\}.
\end{equation*}
This completes the proof of Theorem 1.

\section{Proof of Theorem 3}

For simplicity, let $M_j[h_j[d]]=\sum_{i=1}^n Z_j^{(i)}(d)$. Here,
\begin{align*}
Z_j^{(i)}(d)=\mathbbm{1}_{\{J_i = j\}}\left\{ B(\mathbbm{1}_{\{d^{(i)} =d\}}+\mathbbm{1}_{\{d^{(i)} \neq d\}}\mathbbm{1}_{\{h_j(d^{(i)}) = h_j(d)\}}) \right. \\ \left. +Y \mathbbm{1}_{\{d^{(i)} \neq d\}}\mathbbm{1}_{\{h_j(d^{(i)}) \neq h_j(d)\}}   \right\}, 
\end{align*}
where $B$ and $Y$ are Bernoulli random variables representing the randomized response: $B=1$ with probability $p$ and $B=0$ with probability $1-p$; $Y=1$ with probability $q$ and $Y=0$ with probability $1-q$.

Before we prove Theorem 3, let us first prove several helping lemmas. 

\begin{lem}\label{lem.exp}
$$E[Z_j^{(i)}(d)]=\frac{p}{k}\mathbbm{1}_{\{d^{(i)} =d\}}+\left(\frac{p}{mk}+(1-\frac{1}{m})\frac{q}{k}\right)\mathbbm{1}_{\{d^{(i)} \neq d\}}.$$
\end{lem}

\begin{proof}
By substituting into $\mathbbm{1}_{\{h_j(d^{(i)}) = h_j(d)\}})=\frac{1}{m}$, we can observe that the lemma holds.
\end{proof}

\begin{lem}
$$E\left[ \left(Z_j^{(i)}(d) \right)^2\right]=E[Z_j^{(i)}(d)].$$
\end{lem}

\begin{proof}
  \begin{align*}
 &E\left[ \left(Z_j^{(i)}(d) \right)^2\right] \\
 =&\frac{1}{k}E[B^2 \mathbbm{1}_{\{d^{(i)} =d\}} +\mathbbm{1}_{\{d^{(i)} \neq d\}}(B\mathbbm{1}_{\{h_j(d^{(i)}) = h_j(d)\}}) \\ &+\mathbbm{1}_{\{h_j(d^{(i)}) \neq h_j(d)\}}) Y   )^2 ] \\
 =&\frac{p}{k}\mathbbm{1}_{\{d^{(i)} =d\}}+\left(\frac{p}{mk}+(1-\frac{1}{m})\frac{q}{k}\right)\mathbbm{1}_{\{d^{(i)} \neq d\}}\\
 =&E[Z_j^{(i)}(d)].
  \end{align*}  
\end{proof}

\begin{lem}
For any indices $i_1\neq i_2$, 
\begin{enumerate}
\item $Cov(Z_{j_1}^{(i_1)}(d), Z_{j_2}^{(i_2)}(d)) =0 $ when $j_1\neq j_2$; \\
\item $Cov(Z_{j}^{(i_1)}(d), Z_{j}^{(i_2)}(d)) =0$ when $d^{(i_1)}=d$ or $d^{(i_2)}=d$ or $d^{(i_1)} \neq d^{(i_2)}$; \\
\item $Cov(Z_{j}^{(i_1)}(d), Z_{j}^{(i_2)}(d)) =\left( \frac{p-q}{k} \right)^2 \left(\frac{1}{m}-\frac{1}{m^2} \right) $ when $d^{(i_1)}=d^{(i_2)}=d'$ and $d'\neq d$.
\end{enumerate}
\end{lem}

\begin{proof}
\textbf{Case 1: } When $j_1\neq j_2$, since the hash functions $h_{j_1}$ and $h_{j_2}$ are chosen independently, we have  $Cov(Z_{j_1}^{(i_1)}(d), Z_{j_2}^{(i_2)}(d)) =0 $.  
 
\noindent\textbf{Case 2: } When $j_1=j_2=j$, 
the calculation of covariance depends on the expectation of their product $E[Z_{j_1}^{(i_1)}(d)Z_{j_2}^{(i_2)}(d)]$.
Let $B_1, B_2$ (and $Y_1, Y_2$) denote the randomized response probabilities for $i_1,i_2$, respectively. 
We first consider when $d^{(i_1)}=d$, and we have
\begin{align*}
 &E[Z_{j_1}^{(i_1)}(d)Z_{j_2}^{(i_2)}(d)]  \\
 =& \frac{1}{k^2} E[ B_1(B_2(\mathbbm{1}_{\{d^{(i_2)} =d\}} +\mathbbm{1}_{\{d^{(i_2)} \neq d\}}\mathbbm{1}_{\{h_j(d^{(i_2)}) = h_j(d)\}}) \\  &+Y_2 \mathbbm{1}_{\{d^{(i_2)} \neq d\}}\mathbbm{1}_{\{h_j(d^{(i_2)}) \neq h_j(d)\}}   )  ]  \\
 =& \frac{p}{k^2}\left( p \mathbbm{1}_{\{d^{(i_2)} =d\}} + \left(\frac{p}{m} +(1-\frac{1}{m})q\right) \mathbbm{1}_{\{d^{(i_2)} \neq d\}}    \right)\\
 =& \frac{p}{k} E[Z_j^{(i_2)}(d)] \\
 =& E[Z_j^{(i_1)}(d)]E[Z_j^{(i_2)}(d)].
\end{align*}
Therefore,
\begin{align*}
 &Cov(Z_j^{(i_1)}(d), Z_j^{(i_2)}(d))\\
=&  E[Z_{j_1}^{(i_1)}(d)Z_{j_2}^{(i_2)}(d)]-E[Z_j^{(i_1)}(d)]E[Z_j^{(i_2)}(d)]\\
=& 0.
\end{align*}
Similarly, we also have $Cov(Z_j^{(i_1)}(d), Z_j^{(i_2)}(d))=0$ when $d^{(i_2)}=d$. Then it remains to consider the case when they both are not equal to $d$: $d^{(i_1)}\neq d$ and $d^{(i_2)}\neq d$. We further divide it into two cases.

\noindent\textbf{Case 3: } When $d^{(i_1)}\neq d^{(i_2)}$, all random variables $B_1,B_2,Y_1,Y_2, \mathbbm{1}_{\{d^{(i_2)} \neq d\}}, \mathbbm{1}_{\{h_j(d^{(i_1)}) \neq h_j(d)\}} , \mathbbm{1}_{\{h_j(d^{(i_2)}) \neq h_j(d)\}} $ are independent, so we have $Cov(Z_j^{(i_1)}(d), Z_j^{(i_2)}(d))=0$.

\noindent\textbf{Case 4: } When $d^{(i_1)}= d^{(i_2)}=d'$ and $d'\neq d$, we have 
\begin{align*}
&E[Z_{j_1}^{(i_1)}(d)Z_{j_2}^{(i_2)}(d)]  \\
 =& \frac{1}{k^2}E[ B_1 B_2 \mathbbm{1}^2_{\{h_j(d') = h_j(d)\}}) + Y_1Y_2\mathbbm{1}^2_{\{h_j(d') \neq h_j(d)\}} ]\\
 = &\frac{1}{k^2}\left( \frac{p^2}{m} + q^2(1-\frac{1}{m}\right),
\end{align*}
and by Lemma \ref{lem.exp}, 
\begin{align*}
E[Z_j^{(i_1)}(d)]E[Z_j^{(i_2)}(d)] =\frac{1}{k^2}\left( \frac{p}{m} + q(1-\frac{1}{m})\right)^2.
\end{align*}
Then,
\begin{align*}
&Cov(Z_j^{(i_1)}(d), Z_j^{(i_2)}(d))\\
=&  E[Z_{j_1}^{(i_1)}(d)Z_{j_2}^{(i_2)}(d)]-E[Z_j^{(i_1)}(d)]E[Z_j^{(i_2)}(d)]\\
=&  
\left(\frac{p-q}{k}\right)^2 \frac{1}{m}(1-\frac{1}{m}).
\end{align*}

\end{proof}

We are now ready to prove Theorem 3.
\begin{proof}
By Lemma 2, we have

\begin{align*}
    &E[C(d)]\\ 
    =&\sum_{j=1}^k \sum_{i=1}^n E\left[ Z_j^{(i)}(d)\right] \\
=& f(d)\left(1-\frac{1}{m}\right)(p-q) +\frac{pn}{m} +\left(1-\frac{1}{m}\right)nq.
\end{align*}

Therefore, 
\begin{equation*}
    E[\hat{f}(d)] = \frac{E[C(d)]-\frac{pn}{m}-qn\left(1 -\frac{1}{m}\right)}{(p-q)\left(1 -\frac{1}{m}\right)} = f(d).
\end{equation*}

We can decompose the variance as follows.
\begin{align}
&Var(C(d)) \notag\\  
=&Var\left(\sum_{j=1}^k\sum_{i=1}^nZ_j^{(i)}(d)\right) \notag \\
=&\sum_{i=1}^n Var(\sum_{j=1}^k Z_j^{(i)}(d)) + \sum_{i_1\neq i_2} Cov(Z_{j_1}^{(i_1)}(d), Z_{j_2}^{(i_2)}(d)) \notag \\
=&\sum_{i=1}^n \left(\sum_{j=1}^k Var(Z_j^{(i)}(d))   + \sum_{j_1\neq j_2} Cov(Z_{j_1}^{(i)}(d), Z_{j_2}^{(i)}(d))    \right) \notag \\
&+\sum_{i_1\neq i_2}\sum_{j=1}^k Cov(Z_{j}^{(i_1)}(d), Z_{j}^{(i_2)}(d)), \label{eq.var}
\end{align}
where \eqref{eq.var} follows from Lemma 4(1).

Again, by Lemma 4 (1), we need only consider $\sum_{i=1}^n \sum_{j=1}^k Var(Z_j^{(i)}(d)) $ for the first term in \eqref{eq.var}:
\begin{align}
   &\sum_{i=1}^n \sum_{j=1}^k Var(Z_j^{(i)}(d)) \notag \\
=&\sum_{i=1}^n \sum_{j=1}^k E\left[ \left(Z_j^{(i)}(d) \right)^2\right] - E^2[Z_j^{(i)}(d)] \notag \\
=&pf(d)+\left(\frac{p}{m}+(1-\frac{1}{m})q\right)(n-f(d)) \notag \notag \\
&-\frac{p^2}{k}f(d)+ \frac{1}{k}\left(\frac{p}{m}+(1-\frac{1}{m})q\right)^2(n-f(d)) \label{eq.part1}
\end{align}

Then we turn to consider the second term in \eqref{eq.var}. By Lemma 4 (2) and (3), we need only consider the indices where $d^{(i_1)}=d^{(i_2)}=d^*$ and $d^*\neq d$.
\begin{align}
&\sum_{i_1\neq i_2}\sum_{j=1}^k Cov(Z_{j}^{(i_1)}(d), Z_{j}^{(i_2)}(d)) \notag\\
=&k\left(\frac{p-q}{k}\right)^2 \frac{1}{m}(1-\frac{1}{m})\sum_{d^*\neq d} \sum_{i_1:d^{(i_1)}=d^*}\sum_{i_2\neq i_1} \mathbbm{1}_{\{d^{(i_2)} =d^*\}}   \notag \\
=&k\left(\frac{p-q}{k}\right)^2 \frac{1}{m}(1-\frac{1}{m})\sum_{d^*\neq d} \sum_{i_1:d^{(i)}=d^*} (f(d^*)-1) \notag\\
=&k\left(\frac{p-q}{k}\right)^2 \frac{1}{m}(1-\frac{1}{m})\sum_{d^*\neq d} (f(d^*)-1)f(d^*) \notag\\
=&k\left(\frac{p-q}{k}\right)^2 \frac{1}{m}(1-\frac{1}{m}) \left( \sum_{d^*\neq d}f(d^*)^2-(n-f(d)) \right) \label{eq.part2}
\end{align}
Combining \eqref{eq.part1} and \eqref{eq.part2}, we have
\begin{align*}
&Var(C(d)) \\  
=&f(d)(p-\frac{p^2}{k})+(n-f(d))\left(\frac{p^2}{mk}+\frac{q^2}{k} (1-\frac{1}{m}) \right) \\
&+\frac{(p-q)^2}{km} (1-\frac{1}{m}) \left( \sum_{d^*\neq d}f(d^*)^2 \right).
\end{align*}
By Proposition 1, we have the following relation between $\text{Var}(\hat{f}(d))$  and $Var(C(d))$:
$$\text{Var}(\hat{f}(d))=\frac{Var(C(d))}{(p-q)^2(1-1/m)^2},$$ completing the proof.

\end{proof}

\section{Proof of lemma 1}
In this section, we prove the result in lemma 1, which shows the optimal $p$ can be derived from the following optimization problem:
\begin{equation*}
    p = \text{argmin}\left\{ \frac{kw-p +\lambda{(w-1)(kw-p-wp)}}{k(1-p)^2p}\right\},
\end{equation*}

We start from simplifying variables defined in the expression of the variance.
\begin{equation*}
    \begin{aligned}
\alpha_2 &= \frac{\left( p - \frac{p(m - e^{\epsilon}(1 - p) - p)}{(e^{\epsilon} (1 - p) + p)(m - 1)} \right)^2}{km} \left(1 - \frac{1}{m}\right) \\[10pt]
&=\frac{p^2(w - 1)^2}{w^2k(m-1)},
\end{aligned}
\end{equation*}

and 
\begin{equation*}
\begin{aligned}
    \alpha_1 = &
\frac{p}{m} + \left(1 - \frac{1}{m}\right) \frac{p(m - e^{\epsilon}(1 - p) - p)}{(e^{\epsilon} (1 - p) + p)(m - 1)} \\
&= \frac{p}{m} + \frac{p(m - e^{\epsilon}(1 - p) - p)}{m (e^{\epsilon} (1 - p) + p)} \\[10pt]
&= \frac{p m}{m (e^{\epsilon} (1 - p) + p)} \\[10pt]
&= \frac{p}{e^{\epsilon} (1 - p) + p} = \frac{p}{w}.
\end{aligned}
\end{equation*}

Also 
{
\begin{equation*}
    (p-q)^2 (1-1/m) = \alpha_2 k(m-1)=\frac{p^2(w-1)^2}{w^2}.
\end{equation*}
}
Then minimizing $Var[\hat{f}(d)]$ is equivalent to minimize
\begin{equation}\label{eq:minmize}
\begin{aligned}
\frac{w^2}{(w-1)^2p}\left[f(d)(1-\frac{p}{k}) + (n-f(d)) (\frac{1}{w} - \frac{p}{w^2k})\right]
\end{aligned}
\end{equation}

As $w-1 = (e^{\epsilon}-1)(1-p)$, and $e^{\epsilon}-1 >0$, \eqref{eq:minmize} is equivalent to
\begin{equation*}\label{eq:6}
    \begin{aligned}
    &\frac{w^2}{(1-p)^2p}\left[f(d)(1-\frac{p}{k}) + (n-f(d)) (\frac{1}{w} - \frac{p}{w^2k})\right]\\
    =&\frac{f(d)\left({kw^2 - w^2p}\right) + (n-f(d)) \left({kw-p}\right)}{k(1-p)^2p}\\
    =&\frac{n\left({kw-p}\right)}{k(1-p)^2p} + \frac{f(d)\left({kw^2 - w^2p} - {kw+p}\right)}{k(1-p)^2p}\\
    =&\frac{n\left({kw-p}\right)}{k(1-p)^2p} + \frac{f(d)\left({kw(w-1)+p(1-w)(1+w)}\right)}{k(1-p)^2p}\\
    =&n\left\{\frac{\left({kw-p}\right)}{k(1-p)^2p} + \frac{\lambda{(w-1)(kw-p-wp)}}{k(1-p)^2p}\right\}.
    \end{aligned}
\end{equation*}
This completes the proof for Lemma 1.
\end{document}